\documentclass[11pt,a4paper]{article}
\usepackage{amsmath, amsfonts, amssymb}
\usepackage{a4wide}
\usepackage{parskip}
\usepackage{enumitem}
\usepackage{xcolor}

\usepackage{hyperref}
\usepackage{booktabs}% nicer horizontal lines
\usepackage{caption}% fix vertical spacing of table captions
\usepackage{siunitx}% align numbers at the decimal point
\usepackage{tabularx}
\usepackage{comment}
\usepackage{multirow}
\usepackage{mathtools}
\usepackage{caption}
\usepackage{parskip}
\usepackage{graphicx}
\usepackage{adjustbox}
\usepackage{multirow}
\usepackage{amsthm}
\usepackage{bm, makecell}
\usepackage{lscape}
\usepackage{rotating}
\usepackage{floatrow}
\usepackage[noadjust]{cite}

\newtheorem{theorem}{Theorem}[section]
\newtheorem{lemma}[theorem]{Lemma}
\newtheorem{corollary}[theorem]{Corollary}
\newtheorem{proposition}[theorem]{Proposition}
\theoremstyle{definition}
\newtheorem{definition}[theorem]{Definition}

\newtheorem{notation}[theorem]{Notation}
\newtheorem{example}[theorem]{Example}
\newtheorem{remark}[theorem]{Remark}

\numberwithin{equation}{section}

\usepackage{tikz,xcolor,hyperref}
\usepackage{mathdots}

\definecolor{lime}{HTML}{A6CE39}
\DeclareRobustCommand{\orcidicon}{%
	\begin{tikzpicture}
		\draw[lime, fill=lime] (0,0) 
		circle [radius=0.16] 
		node[white] {{\fontfamily{qag}\selectfont \tiny ID}};
		\draw[white, fill=white] (-0.0625,0.095) 
		circle [radius=0.007];
	\end{tikzpicture}
	\hspace{-2mm}
}

\foreach \x in {A, ..., Z}{%
	\expandafter\xdef\csname orcid\x\endcsname{\noexpand\href{https://orcid.org/\csname orcidauthor\x\endcsname}{\noexpand\orcidicon}}
}

\begin{document}
	\date{}
	%\vspace{0.01in}
		\title{Polycyclic Codes over the Product Ring $\mathbb{F}_q^l$ and their Annihilator Dual  }
		\author{{\bf Akanksha\footnote{email: {\tt akankshafzd8@gmail.com}}\orcidA{}, \bf Ritumoni Sarma\footnote{	email: {\tt ritumoni407@gmail.com}}\orcidC{}} \\ Department of Mathematics\\ Indian Institute of Technology Delhi\\Hauz Khas, New Delhi-110016, India}
  
\maketitle

\begin{abstract}
In this article, for the finite field $\mathbb{F}_q$, we show that the $\mathbb{F}_q$-algebra $\mathbb{F}_q[x]/\langle f(x) \rangle$ is isomorphic to the product ring $\mathbb{F}_q^{\deg f(x)}$ if and only if $f(x)$ splits over $\mathbb{F}_q$ into distinct factors. We generalize this result to the quotient of the polynomial algebra $\mathbb{F}_q[x_1, x_2,\dots, x_k]$ by the ideal $\langle f_1(x_1), f_2(x_2),\dots, f_k(x_k)\rangle.$
On the other hand, we establish that every finite-dimensional $\mathbb{F}_q$-algebra $\mathcal{S}$ has an orthogonal basis of idempotents with their sum equal to $1_{\mathcal{S}}$ if and only if $\mathcal{S}\cong\mathbb{F}_q^l$ as $\mathbb{F}_q$-algebras, where  $l=\dim_{\mathbb{F}_q} \mathcal{S}$. Instead of studying polycyclic codes over $\mathbb{F}_q$-algebras $\mathbb{F}_q[x_1, x_2,\dots, x_k]/\langle  f_1(x_1), f_2(x_2),\dots, f_k(x_k)\rangle$ where $f_i(x_i)$ splits into distinct linear factors over $\mathbb{F}_q,$ which is a subclass of $\mathbb{F}_q^l,$  
we study polycyclic codes over $\mathbb{F}_q^l$ and obtain their unique decomposition into polycyclic codes over $\mathbb{F}_q$ for every such orthogonal basis of $\mathbb{F}_q^l$. We refer to it as an $\mathbb{F}_q$-decomposition. An $\mathbb{F}_q$-decomposition enables us to use results of polycyclic codes over $\mathbb{F}_q$ to study polycyclic codes over $\mathbb{F}_q^l$; for instance, we show that the annihilator dual of a polycyclic code over $\mathbb{F}_q^l$ is a polycyclic code over $\mathbb{F}_q^l$. Furthermore, with the help of different Gray maps, we produce a good number of examples of MDS or almost-MDS or/and optimal codes; some of them are LCD over $\mathbb{F}_q$. Finally, we study Gray maps from $(\mathbb{F}_q^l)^n$ to $\mathbb{F}_q^{nl},$ and use it to construct quantum codes with the help of CSS construction.
\medskip

			
\medskip
			
%%%add for cyclic codes and verify all

\end{abstract}

\section{Introduction}\label{Section 1}\label{sec1}
% Throughout this article, $\mathbb{F}_q$ denotes the finite field of cardinality $q=p^m,$ for a prime $p$ and a natural number $m$. A \textit{linear code} $C$ is an $\mathbb{F}_q$-subspace of $\mathbb{F}_q^n;$ the number $n$ is called the \textit{length} of $C$. 
Linear codes have been extensively studied over finite fields due to their crucial role in error detection and error correction, which are essential for reliable information transmission. In 1994, Hammons et al. (\cite{hammons1994z}) studied linear codes over $\mathbb{Z}_4$ and constructed codes as binary images under the Gray maps of linear codes over $\mathbb{Z}_4$ for the first time. This drew the attention of coding theorists towards codes over finite commutative rings. 
% For a finite commutative ring $R$ with unity, an $R$-linear code of length $n$ is an $R$-submodule of $R^n$.
A linear code invariant under the right (and hence left) cyclic shift is called a \textit{cyclic code},
first introduced by E. Prange (\cite{prange1957cyclic}) in 1957. Due to their relatively easier implementation and a great deal of rich algebraic structure, they have been studied widely. Numerous algebraic coding theorists have encouraged the study of cyclic codes for both burst-error correction and random-error correction. Cyclic codes over finite chain rings have been studied in different contexts, for instance, see \cite{abualrub2007cyclic},\cite{bonnecaze1999cyclic},\cite{dinh2004cyclic} and \cite{kanwar1997cyclic} and many authors also studied cyclic codes over finite non-chain rings, for instance,  
see \cite{yildiz2011cyclic} and \cite{islam2021cyclic}.
\par

    Constacyclic codes are a generalization of cyclic codes, first introduced by
    Berlekamp (\cite{berlekamp2015algebraic}) in 1960. Negacyclic codes are a special case of constacyclic codes. Both constacyclic and negacyclic codes over finite fields have been widely studied, for example, see \cite{bakshi2012class},\cite{chen2012constacyclic} and \cite{raka2015class}. Further, constacyclic codes over various rings have been extensively studied, for instance, see \cite{ABUALRUB2009520},\cite{swati_raka}, \cite{gao2018u},\cite{karadeniz20111+},\cite{QIAN2006820}, \cite{raza2024quantum} and \cite{shi2021construction}. Note that the rings considered to study constacyclic codes in \cite{swati_raka}, \cite{gao2018u}, \cite{raza2024quantum} and  \cite{shi2021construction} are isomorphic to a product ring of the form $\mathbb{F}_q^l,$ for some $l \in \mathbb{N}.$ 
    \par
    Polycyclic codes are a generalization of constacyclic codes, introduced by Peterson (\cite{peterson1972error}) in 1972. In recent years, polycyclic codes have been studied in \cite{aydin2022polycyclic},\cite{lopez2009dual},\cite{shi2020polycyclic},\cite{shi2020construction}, \cite{tabue2018polycyclic} and \cite{wu2022structure}. In 2009, L\'{o}pez-Permouth et al. (\cite{lopez2009dual}) studied polycyclic codes and sequential codes, and they established that a linear code over $\mathbb{F}_q$ is polycyclic if and only if its Euclidean dual is sequential. However, sequential codes are not necessarily polycyclic, and hence, the Euclidean dual of a polycyclic code need not be polycyclic. In 2016, Alahmadi et al. (\cite{alahmadi2016duality}) studied polycyclic codes over finite fields, and they introduced a different non-degenerate bilinear form, with respect to which the dual of polycyclic code is polycyclic. In 2020, Tabue et al. (\cite{tabue2018polycyclic}) studied polycyclic codes over finite chain rings and presented several results regarding the freeness of polycyclic codes and their annihilator dual. 
\par
In 2021, Islam et al. (\cite{islam2021cyclic}) studied cyclic codes over $R_{e,q}:=\mathbb{F}_q[u]/\langle u^e-1 \rangle,$ where $e|(q-1).$ In 2022, Bhardwaj and Raka (\cite{swati_raka}) studied constacyclic codes over a general non-chain ring $\mathcal{R}:=$ $\mathbb{F}_q[u_1,u_2,\dots,u_k]/\langle f_1(u_1),f_2(u_2),\dots,f_k(u_k)\rangle,$ where each $f_i(u_i)$ splits over $\mathbb{F}_q$ and has simple roots. Observe that both $R_{e, q}$ and $\mathcal{R}$ are isomorphic to a product ring $\mathbb{F}_q^l, $ for some $l \in \mathbb{N}.$ In 2021, Qi (\cite{qi2022polycyclic}) studied polycyclic codes over $\mathbb{F}_q[u]/\langle u^2-u\rangle,$ a non-chain ring which is isomorphic to $\mathbb{F}_q^2.$ However, note that in any case $\mathcal{R}$ and $R_{e,q}$ are not isomorphic to $\mathbb{F}_{p_1}^{p_2},$ where $p_1<p_2$ are primes. Hence, it is worth mentioning that the product ring $\mathbb{F}_q^l$ is a wider class as compared to the rings considered in the literature mentioned above. In 2024, Bajalan and Moro (\cite{bajalan2024polycyclic}) studied polycyclic codes over rings of the form $\mathcal{A}:=R[x_1,x_2,\dots,x_k]/\langle t_1(x_1),t_2(x_2),\dots,t_k(x_k)\rangle,$ where $R$ denotes a finite chain ring and each $t_i(x_i)$ is a monic square-free polynomial over $R.$ It is important to emphasize that the base ring considered in our article does not represent a special case of the class considered by the authors in \cite{bajalan2024polycyclic}. Moreover, neither of the two classes can be regarded as a generalization of the other. Specifically, for any choice of a finite chain ring $R,$ $\mathcal{A}$ is not isomorphic to $\mathbb{F}_{p_1}^{p_2},$ where $p_1< p_2$ are primes. This motivates us to study polycyclic codes over the product ring $\mathbb{F}_q^l,$ for some $l \in \mathbb{N}.$
\par
The structure of the article is as follows: Preliminaries are presented in Section \ref{sec2}. In Section \ref{sec3}, we characterize a ring $\mathcal{S}$ isomorphic to a product ring $\mathbb{F}_q^l.$ In Section \ref{sec4}, we study polycyclic codes over $\mathbb{F}_q^l$ and study their annihilator duals. In Section \ref{sec5}, we consider different Gray maps from the $\mathbb{F}_q$-algebra $(\mathbb{F}_q^l)^n$ to $\mathbb{F}_q^{nl}$ and study certain properties of the Gray image. Moreover, we utilize Gray maps from $(\mathbb{F}_q^l)^n$ to $\mathbb{F}_q^{nl}$ to construct quantum codes with the help of CSS construction. We conclude the article in Section \ref{sec6}.  
\section{Preliminaries}\label{sec2}
Throughout this article, for a prime power $q=p^m,$ $\mathbb{F}_q$ denotes the finite field of order $q.$
Let $R$ denote a finite commutative unital ring, and $U(R)$ denote the group of units of $R$.
\begin{definition}[$R$-linear Code] A code $C$ of length $n$ is said to be \textit{$R$-linear code}, if it is an $R$-submodule of $R^n.$
\end{definition}
Recall that, for an $\mathbb{F}_q$-linear code $C$ of length $n$, dimension $k$ and distance $d$ we denote its parameters by $[n,k,d].$ Every $\mathbb{F}_q$-linear code with parameters $[n,k,d]$ satisfies $d -1\le n-k$, and codes for which $d-1=n-k,$ are said to be \textit{maximum distance separable} (in short, \textit{MDS}) codes.
These codes maximize the distance of the code. Consequently, they can detect and correct the highest number of errors.
\begin{definition}[Euclidean Dual]
    The \textit{Euclidean dual} of an $R$-linear code $C$ of length $n$, denoted by $C^\perp,$ is the $R$-linear code given by $\{\mathbf{x}\in R^n : \langle \mathbf{x}, \mathbf{y}\rangle =\mathbf{x}\cdot \mathbf{y}=\underset{i=1}{\overset{n}{\sum}}x_iy_i=0 \, , \forall \, \mathbf{y} \in C\}.$
\end{definition}

\begin{definition}[LCD Codes]
An {\it LCD code} $C$ is an $R$-linear code which intersects its Euclidean dual trivially, that is, $C\cap C^{\perp}=\{0\}$.  
\end{definition}
\begin{definition}
    Let  $C$ be an $R$-linear code,
    \begin{itemize}
        \item It is {\it self-orthogonal} if $ C \subset C^\perp$,
        \item It is {\it self-dual} if $C=C^\perp,$
        \item It is {\it dual-containing} if $C^\perp\subset C.$
    \end{itemize}
\end{definition}
\begin{definition}[Polycyclic Codes]
Suppose that $ \mathbf{a} =(a_0,a_1,\dots,a_{n-1}) \in R^n $ and $a_0 \in U(R)$. A linear code $C$ over $R$ is called \textit{$\mathbf{a}$-polycyclic} if whenever $ \mathbf{c}=(c_0,c_1,\dots,c_{n-1}) \in C$, then $(0,c_0,c_1,\dots,c_{n-2})+c_{n-1}(a_0,a_1,\dots,a_{n-1}) \in C.$
\end{definition}
Associate
\begin{equation}\label{tuple to polynomial}
\mathbf{b}=(b_0,b_1,\dots,b_{n-1}) \in R^n \ \text{ with } \ \mathbf{b}(x)=b_0+b_1x+\dots+b_{n-1}x^{n-1} \in R[x].
\end{equation}
Similar to a cyclic code, for $\mathbf{a}=(a_0,a_1,\dots,a_{n-1})\in R^n,$ we can characterize an $\mathbf{a}$-polycyclic code over $R$ as an ideal of $R^{\mathbf{a}} =R[x]/\langle x^n-\mathbf{a}(x)\rangle$, for $\mathbf{a}(x)\in R[x]$. 
\begin{proposition}\cite{lopez2009dual}
A linear code over $R$ is {\it $\mathbf{a}$-polycyclic code} if and only if it is an ideal of $R^\mathbf{a}$. 
\end{proposition}

\begin{definition}[Sequential Codes]
An $R$-linear code $C$ of length $n$ is said to be {\it $\textbf{b}$-sequential}, for $ \textbf{b} \in R^n $, if $ \mathbf{c}=(c_0,c_1,\dots,c_{n-1}) \in C\Longrightarrow (c_1,c_2,\dots, c_{n-1}, \mathbf{c}\cdot\textbf{b}) \in C,$ where $\mathbf{c}\cdot\mathbf{b}$ denotes the Euclidean inner product of $\mathbf{c}$ and $\mathbf{b}$ in $R^n.$ 
\end{definition}
\begin{definition}($\mathbb{F}_q$-algebra)  \cite{atiyah2018introduction}
A commutative ring $R$ with unity (denoted by $1_R$) is called an $\mathbb{F}_q$-algebra if there exists a ring homomorphism $f:\mathbb{F}_q\to R$ such that $f(1)=1_R.$
    
\end{definition}

\begin{notation}[Product Ring] The $l$-dimensional vector space $\mathbb{F}_q^l$ is a ring with respect to the component-wise multiplication, that is, for $\mathbf{x}=(x_1, x_2, \dots, x_l), \mathbf{y}=(y_1, y_2,\dots, y_l)\in \mathbb{F}_q^l,$ the multiplication is $\mathbf{x} \mathbf{y}=(x_1y_1, x_2y_2,\dots, x_ly_l).$ We denote this product ring simply by $\mathbb{F}_q^l.$  Note that for $a\in\mathbb{F}_q$, $a\mapsto (a, a, \dots, a)$ defines an $\mathbb{F}_q$-algebra structure on this product ring.  
\end{notation}
\section{Characterization of the Base Ring}\label{sec3}
In this section we investigate when a particular class of $\mathbb{F}_q$-algebra is isomorphic to a product ring $\mathbb{F}_q^l$.
% Denote by $\mathbb{F}_q[x]$ the ring of polynomials in $x$ over $\mathbb{F}_q.$
\begin{lemma}\label{Existence of basis consisting of orthogonal idempotents}
Let $\mathcal{S}$ be an $\mathbb{F}_q$-algebra. Then, $\mathcal{S}$ is isomorphic to the $\mathbb{F}_q$-algebra $\mathbb{F}_q^l$ if and only if there exists an $\mathbb{F}_q$-basis $\{e_1, e_2, \dots, e_l\}$ of $\mathcal{S}$ consisting of orthogonal idempotent elements satisfying $\sum\limits_{i=1}^l{e_i}=1$.
\end{lemma}
\begin{proof}For $1\le i\le l$, denote by $\epsilon_i,$ the $l$-tuple over $\mathbb{F}_q$ whose $i$-th component is $1$ and $j$-th component is $0$ for $1\le j\ne i\le l$ so that $\{\epsilon_1,\dots,\epsilon_l\}$ is a basis of $\mathbb{F}_q^l.$
    Let $\phi:\mathcal{S} \to \mathbb{F}_q^l$ be an isomorphism of $\mathbb{F}_q$-algebras. Then, $\{\phi^{-1}(\epsilon_1),\dots,\phi^{-1}(\epsilon_l)\}$ is an $\mathbb{F}_q$-basis of $\mathcal{S}$ consisting of orthogonal idempotents such that $\underset{i=1}{\overset{l}{\sum}}\phi^{-1}(\epsilon_i)=1.$
    Conversely, let $\mathcal{S}$ be an $\mathbb{F}_q$-algebra that has a basis $\{e_1,\dots,e_l\}$ such that $e_ie_j=\delta_{i, j} e_i$ and $\underset{i=1}{\overset{l}{\sum}} e_i=1.$ Then, $e_i \mapsto \epsilon_i$, for $1\le i\le l$,  induces  an isomorphism of $\mathbb{F}_q$-algebras from $\mathcal{S}$ to $\mathbb{F}_q^l$.\hfill $\square$
\end{proof}
%For an $\mathbb{F}_q$-algebra $\mathbb{F}_q^l,$ isomorphic to a product ring, from lemma \ref{Existence of basis consisting of orthogonal idempotents} we can find a basis of $\mathbb{F}_q^l$ over $\mathbb{F}_q$ consisting of orthogonal idempotents having sum $1_{\mathbb{F}_q^l},$ and this basis has the properties in the following proposition.
\begin{proposition}
If an $\mathbb{F}_q$-algebra $\mathcal{S}$ has an $\mathbb{F}_q$-basis $\mathcal{B}=\{e_1, e_2, \dots, e_l\}$ consisting of orthogonal idempotent elements satisfying $\sum\limits_{i=1}^l{e_i}=1$, then we have the following:  
\begin{enumerate}
    \item $e_i\mathbb{F}_q\lhd \mathcal{S}$ and $\mathcal{S}=\underset{i=1}{\overset{l}{\bigoplus}}e_i\mathbb{F}_q$;
    \item $e_i\mathcal{S}=e_i\mathbb{F}_q,$ and hence $\mathcal{S}=\underset{i=1}{\overset{l}{\bigoplus}}e_i\mathcal{S}.$
\end{enumerate}
\end{proposition}
\begin{proof}  Suppose $a=\underset{j=1}{\overset{l}{\sum}}e_ja_j\in\mathcal{S}$ for $a_j\in\mathbb{F}_q.$
\begin{enumerate} 
    \item Observe $e_i\mathbb{F}_q<\mathcal{S}$. Since $e_ie_j=0$ and $e_i^2=e_i$, we have $a(e_ib)=e_ia_ib$ for $b\in\mathbb{F}_q$ so that $e_i\mathbb{F}_q\lhd\mathcal{S}$. Now the next statement follows as $\mathcal{B}$ is a basis.
    \item Observe that $e_i\mathcal{S}\supset e_i\mathbb{F}_q.$ For the other containment, observe $e_ia=e_ia_i.$  
\end{enumerate} \hfill $\square$
\end{proof}
 \begin{lemma}\label{Rings}
 Let $e, l\in\mathbb{N}$ and $p(x)\in \mathbb{F}_q[x]$ be irreducible. Then, $\frac{\mathbb{F}_q[x]}{\langle{p(x)}^{e}\rangle}$ is isomorphic to the product ring ${\mathbb{F}_q^{l}}$ if and only if $\deg(p(x))=1$, $e=1$ and $l=1$.
 \end{lemma}
 \begin{proof}Assume that $\frac{\mathbb{F}_q[x]}{\langle{p(x)}^{e}\rangle}\cong{\mathbb{F}_q^{l}}$ as rings. %Consider the $\mathbb{F}_q$ basis of $\frac{\mathbb{F}_q[x]}{\langle{p(x)}^{e}\rangle}$ given by $\{x^ip(x)^j:0\leq i \leq{\deg{p(x)}-1}, 0\leq j \leq{e-1}\}.$ 
     Since $p(x)$ is irreducible over $\mathbb{F}_q,$ the group of units of $\frac{\mathbb{F}_q[x]}{\langle{p(x)}^{e}\rangle}$ is the complement of $\langle p(x)\rangle/\langle p(x)^e\rangle$ so that the number of units in $\frac{\mathbb{F}_q[x]}{\langle{p(x)}^{e}\rangle}$ is $q^{e\deg{p(x)}}-q^{\deg{p(x)}(e-1)}=q^{(e-1)\deg{p(x)}}(q^{\deg{p(x)}}-1).$ But the number of units in $\mathbb{F}_q^{l}$ is $(q-1)^l.$ Hence, $q^{(e-1)\deg{p(x)}}(q^{\deg{p(x)}}-1)=(q-1)^l.$ As $ 1\leq \deg{p(x)},$ we get $e=1, \,\deg{p(x)}=1$ and $l=1.$\\ 
     Conversely, if $\deg p(x)=e=l=1$, then $p(x)=x+a$ for some $a\in\mathbb{F}_q$ and $\frac{\mathbb{F}_q[x]}{\langle{x+a}\rangle}\cong \mathbb{F}_q$.\hfill $\square$
 \end{proof}
The following theorem is immediate from Lemma \ref{Rings}.

\begin{theorem}\label{isomorphic to F_q^l}
Let $f(x)\in\mathbb{F}_q[x]$, $l\in\mathbb{N}$.Then, $\frac{\mathbb{F}_q[x]}{\langle f(x) \rangle}$ is isomorphic to the product ring ${\mathbb{F}_q^l}$ if and only if $f(x)$ splits over $\mathbb{F}_q$ into distinct linear factors and $l=\deg{f(x)}$.
\end{theorem}
\begin{proof}Suppose the factorization of $f(x)\in\mathbb{F}_q[x]$ into irreducible factors is given by $$p_1(x)^{e_1}p_2(x)^{e_2}\dots p_r(x)^{e_r}$$ where $p_i(x)\ne p_j(x)$ for $i\ne j.$ Then, by the Chinese Remainder Theorem, $\frac{\mathbb{F}_q[x]}{\langle f(x) \rangle} \cong \underset{i=1}{\overset{r}{\prod}}\frac{\mathbb{F}_q[x]}{\langle p_i(x)^{e_i} \rangle}.$ By Lemma \ref{Rings}, it is isomorphic to $\mathbb{F}_q^l$ if and only if each factor is isomorphic to $\mathbb{F}_q$ so that $\deg{p_i(x)}=1$ and $ e_i=1$ for each $i.$ \hfill $\square$
\end{proof}
For a field $\mathbb{F}$, if $A$ and $B$ are $\mathbb{F}$-algebras, then the tensor product $A\otimes_\mathbb{F} B$ (or $A\otimes B$) is an $\mathbb{F}$-algebra [see chapter 2, \cite{atiyah2018introduction}].

\begin{lemma}\label{tensortodirectproduct} Let $\mathbb{F}$ be a field and $m, n\in\mathbb{N}$. If $A=\mathbb{F}^m$ and $B=\mathbb{F}^n,$ then $A\otimes B\cong \mathbb{F}^{mn}$ as $\mathbb{F}$-algebras. More generally, $ \underset{i=1}{\overset{l}{\otimes}}\mathbb{F}^{d_i}\cong \mathbb{F}^{\underset{i=1}{\overset{l}{\prod}}d_i}$  as $\mathbb{F}$-algebras, where $d_i\in \mathbb{N} $  $ \forall \, \,1\leq i\leq  l.$
\end{lemma}
\begin{proof} 
 Let $\{e_1,e_2,\dots,e_{m}\}$ denote the standard base of $A$,  $\{f_1,f_2,\dots,f_{n}\}$ denote the standard base of $B$, and let $\{\epsilon_1,\epsilon_2,\dots,\epsilon_{mn}\}$ denote the standard base of ${\mathbb{F}^{mn}}$. Consider the bilinear map $\psi: \mathbb{F}^{m} \times {\mathbb{F}}^{n}\longrightarrow\mathbb{F}^{mn}$ given by $\psi(e_i, f_j)=\epsilon_{i+(j-1)m}.$ Since the image of $\psi$ contains a basis, $\psi$ extends uniquely to an $\mathbb{F}$-linear surjective map $\tilde{\psi}:A\otimes B\to \mathbb{F}^{mn}$ which is in fact an isomorphism of vector spaces as the domain and the codomain spaces have the same dimension over $\mathbb{F}$. For  $\tilde{\psi}$ to be an isomorphism of $\mathbb{F}$-algebras, it is enough to check that $\tilde{\psi}((e_i\otimes f_j)(e_s\otimes f_t))= \tilde{\psi}(e_i\otimes f_j)\tilde{\psi}(e_s\otimes f_t)$ for $1\le i, s\le m, 1\le j, t\le n$. Note that, $\tilde{\psi}((e_i\otimes f_j)(e_s\otimes f_t))=\tilde{\psi}(e_i e_s \otimes  f_j f_t)= \tilde{\psi}(e_i\delta_{i,s} \otimes  f_j\delta_{j,t})=\tilde{\psi}(\delta_{i,s}\delta_{j,t}e_i\otimes f_j)=\delta_{i,s}\delta_{j,t}\epsilon_{i+(j-1)m}.$
On the other hand,   
$\tilde{\psi}(e_i\otimes f_j)\tilde{\psi}(e_s\otimes f_t)=
\epsilon_{i+(j-1)m}\epsilon_{s+(t-1) m}=\delta_{i+(j-1)m, s+(t-1)m}\epsilon_{s+(t-1)m}.$ Observe that $i+(j-1)m=s+(t-1)m$ if and only if $i=s$ and $j=t$. The general statement follows by induction. \hfill $\square$
\end{proof}

In \cite{swati_raka}, Bhardwaj and Raka explicitly constructed a basis consisting of orthogonal idempotents whose sum is the unity (referred to as a complete set of orthogonal idempotents) for the quotient ring 
$\mathcal{R}=\mathbb{F}_q[x_1,\dots,x_k]/\langle f_1(x_1),\dots,f_k(x_k) \rangle,$ where each $f_i(x_i)$ splits over $\mathbb{F}_q$ and has simple roots. As a consequence, they obtained an $\mathbb{F}_q$-decomposition of $\mathcal{R}$
so that $\mathcal{R}\cong \mathbb{F}_q^l$, where $l=\dim_{\mathbb{F}_q} \mathcal{R}$. However, from Theorem \ref{complete orthogonal idempotents}, we see that it is not necessary to explicitly construct a basis consisting of orthogonal idempotents having sum $1,$ to show that $\mathcal{R}\cong\mathbb{F}_q^l$, for some $l\in\mathbb{N}$. Note that one requires the knowledge of roots of each $f_i(x_i)$ to construct such a basis of $\mathbb{F}_q[x_1,\dots,x_k]/\langle f_1(x_1),\dots,f_k(x_k) \rangle,$ where each $f_i(x_i)$ splits over $\mathbb{F}_q$ and has simple roots, by the construction given in \cite{swati_raka}. In fact, the following theorem can also be deduced from the results of the article \cite{poli1985important}, \cite{martinez2006multivariable} or \cite{chillag1995regular}, however, a simple and alternate proof can be seen by the argument given below which follows from the lemma proved above.

\begin{theorem}\label{complete orthogonal idempotents} Suppose, for $1\le i\le k,$  $f_i(x_i)$ splits over $\mathbb{F}_q$ and has simple roots. Then 
     $\mathcal{R}=\mathbb{F}_q[x_1,\dots,x_k]/\langle f_1(x_1),  \dots,f_k(x_k) \rangle$  is isomorphic to a product ring of the form $\mathbb{F}_q^l.$
\end{theorem}
\begin{proof}
    Observe that as rings $$\mathcal{R} \cong \mathbb{F}_q[x_1]/\langle f_1(x_1) \rangle \otimes \mathbb{F}_q[x_2]/\langle f_2(x_2) \rangle \otimes \dots \otimes \mathbb{F}_q[x_k]/\langle f_k(x_k) \rangle.$$
    So, by Theorem \ref{isomorphic to F_q^l}, $$\mathcal{R}\cong
    \mathbb{F}_q^{\deg{f_1}} \otimes \mathbb{F}_q^{\deg{f_2}}\otimes \dots \otimes \mathbb{F}_q^{\deg{f_k}}.$$ 
    Therefore, by Lemma \ref{tensortodirectproduct}, $\mathcal{R}$ is a ring isomorphic to $\mathbb{F}_q^l,$ where $l=\underset{i=1}{\overset{k}{\prod}}\deg{f}_i.$ \hfill $\square$
\end{proof}

\begin{remark} 
    Bhardwaj and Raka in \cite{swati_raka} considered codes over a wide class of rings which were isomorphic to product rings, still, considering codes over a product ring ${\mathbb{F}_q^l}$, where $l \in \mathbb{N},$ is more general than considering codes over a ring of the form  $\mathbb{F}_q[x_1,\dots,x_k]/\langle f_1(x_1),\dots,f_k(x_k) \rangle$ where each $f_i(x_i)$ splits over $\mathbb{F}_q$ and has simple roots and $k\in \mathbb{N}$. Since for a prime power $q$ and $f_1(x_1),\dots,f_k(x_k)\in\mathbb{F}_q[x_1,\dots,x_k],$ we have $\mathbb{F}_q[x_1,\dots,x_k]/\langle f_1(x_1),\dots,f_k(x_k) \rangle\not\cong\mathbb{F}_{p_1}^{p_2},$ where $p_1<p_2$ are primes.
\end{remark}
% {\color{blue} \begin{remark}
%     In \cite{bajalan2024polycyclic} Bajalan and Martinez-Moro studied polycyclic codes over $\frac{R[x_1,\dots,x_s]}{\langle t_1(x_1),\dots,t_s(x_s)\rangle} $, where $R$ is a chain ring and each $t_i(x_i)$ in $R[x_i]$ for $i\in\{1,\dots,s\}$ is a monic square-free polynomial. He further mentioned the case when $R$ is a finite field. However, note that for any choice of square-free polynomials in $\mathbb{F}_2[x_i]$ for $1\leq i\leq s$ and a natural number $s\in \mathbb{N}$ the $\mathbb{F}_2$-algebra $\frac{\mathbb{F}_2[x_1,\dots,x_s]}{\langle t_1(x_1),\dots,t_s(x_s)\rangle }$is not isomorphic to $\mathbb{F}_2^5.$ Note these two classes of rings (product rings and the rings considered by Bajalan and Martinez-Moro in \cite{bajalan2024polycyclic}) considered as base rings are neither generalization of the other. Moreover, it is easy to implement the class of product rings in  MAGMA and sagemath and it is easier to define Gray maps on product rings because they are nothing but simply multiplication by a suitable invertible matrix.
% \end{remark}}
\begin{remark}
    Another advantage of using $\mathbb{F}_q^l$ as the base ring is that a basis consisting of orthogonal idempotents is readily available, namely, the standard basis.
\end{remark}
%\begin{remark}
     %   Instead of studying polycyclic codes over $\frac{\mathbb{F}_q[x]}{\langle f(x) \rangle},$ where $f(x)$ splits over $\mathbb{F}_q$ and has distinct linear factors, it is enough (in fact, more general) to study polycyclic codes over the product ring $\mathbb{F}_q^l.$
   % \end{remark}
  \begin{remark}  
    In a finite commutative $\mathbb{F}_q$-algebra with unity, a complete orthogonal basis of idempotents need not be unique, for instance, $\mathbb{F}_2^4$ has two such bases, namely, \\
        $\{(1,0,0,0),(0,1,0,0),(0,0,1,0),(0,0,0,1)\},$ and 
        $\{(1,1,1,0),(0,1,1,1),(1,0,1,1),(1,1,0,1)\}$. 
   \end{remark}
\begin{remark}\label{remark2.10PIR}
Note that if $\mathcal{S}$ is isomorphic to the product ring $\mathbb{F}_q^l$ then $\mathcal{S}[x]$ is a principal ideal ring since $\mathbb{F}_q^l[x]\cong \mathbb{F}_q[x]\times \dots \times \mathbb{F}_q[x]$($l$ times) as rings, and $\mathbb{F}_q[x]$ is a principal ideal domain. Hence for any polynomial $f(x) \in \mathcal{S}[x]$, the quotient ring $\frac{\mathcal{S}[x]}{\langle f(x) \rangle}$ is also a principal ideal ring.
\end{remark}

\section{Polycyclic Codes over a Product Ring}\label{sec4}
Let $\mathcal{B}=\{\bm{e}_1,\bm{e}_2,\dots,\bm{e}_l\}$ denote an $\mathbb{F}_q$-basis of $\mathbb{F}_q^l$ consisting of orthogonal idempotents having sum $\bm{1}.$ Note that one such basis is the standard basis. Let $\bm{a}\in\mathbb{F}_q^l$ and write $\bm{a}=\underset{i=1}{\overset{l}{\sum}}\bm{e}_ia_i,$ where $a_i\in\mathbb{F}_q,$ for $1\le i\le l$.
Denote the $i$-th canonical projection of $\mathbb{F}_q^l$ to $\mathbb{F}_q$ by $\pi_i$ so 
that $\pi_i(\bm{a})=a_i$. Extend this to $\tilde{\pi}_i:(\mathbb{F}_q^l)^n\to\mathbb{F}_q^n$ such that $\tilde{\pi}_i(\bm{a}_0,\dots,\bm{a}_{n-1})=(\pi_i(\bm{a}_0),\dots,\pi_i(\bm{a}_{n-1})).$ We will use the following notations throughout the article:
\begin{notation} For any $\Bar{\bm{a}}=\left(\bm{a}_0, \bm{a}_2,\dots, \bm{a}_{n-1}\right)\in (\mathbb{F}_q^l)^n,$ if $\bm{a}_i=\underset{j=1}{\overset{l}{\sum}} \bm{e}_ja_{i,j}$, where $a_{i,j}\in\mathbb{F}_q,$ for $0\leq i\leq {n-1}.$  For $1 \leq j \leq l,$ we define, $\bm{a}^{(j)}=\big(a_{0,j},a_{1,j}, \dots, a_{{n-1},j}\big)\in\mathbb{F}_q^n$ so that, \begin{equation}\label{decomposition of a vector}
    \Bar{\bm{a}} =\underset{j=1}{\overset{l}{\sum}}(a_{0,j}\bm{e}_j,a_{1,j}\bm{e}_j,\dots,a_{n-1,j}\bm{e}_j)=\underset{j=1}{\overset{l}{\sum}}\bm{e}_j*{\bm{a}}^{(j)}.
\end{equation}
\end{notation}
\begin{notation}
Let $\mathcal{C}$ be an $\Bar{\bm{a}}$-polycyclic code over $\mathbb{F}_q^l$ for some $\Bar{\bm{a}}=(\bm{a}_0,\bm{a}_1,\dots,\bm{a}_{n-1})\in (\mathbb{F}_q^l)^n.$ Then from Equation \ref{tuple to polynomial} we have $\Bar{\bm{a}}(x)=\bm{a}_0+\bm{a}_1x+\dots+\bm{a}_{n-1}x^{n-1}$ and $\mathcal{C}$ is an ideal of $\mathbb{F}_q^l[x]/\langle x^n-\Bar{\bm{a}}(x)\rangle.$ Further, from Equation \ref{decomposition of a vector} we get 
$\Bar{\bm{a}}=\underset{j=1}{\overset{l}{\sum}}(a_{0,j}\bm{e}_j,a_{1,j}\bm{e}_j,\dots,a_{n-1,j}\bm{e}_j)=\underset{j=1}{\overset{l}{\sum}}\bm{e}_j*{\bm{a}}^{(j)}\in (\mathbb{F}_q^l)^n$ and $\bm{a}^{(j)}=\big(a_{0,j},a_{1,j}, \dots, a_{{n-1},j}\big)\in\mathbb{F}_q^n.$ Observe that, $x^n-\Bar{\bm{a}}(x)=\underset{j=1}{\overset{l}{\sum}}\bm{e}_j(x^n-\bm{a}^{(j)}(x)),$ where $x^n-\bm{a}^{(j)}(x)\in \mathbb{F}_q[x]$ exist uniquely.
\end{notation}
Then,  with respect to the basis $\mathcal{B}$, we have a unique decomposition of every $\mathbb{F}_q^l$-submodule of $(\mathbb{F}_q^l)^n$ into $\mathbb{F}_q$-linear subspaces of $\mathbb{F}_q^n.$  

\begin{theorem}[$\mathbb{F}_q$-decomposition]\label{linear}
    Let $\{\bm{e}_1,\bm{e}_2,\dots,\bm{e}_l\}$ be a basis of $\mathbb{F}_q^l$ over $\mathbb{F}_q$ consisting of orthogonal idempotents such that $\sum_{i=1}^l\bm{e}_i=\bm{1}$. For  $\mathcal{C}\subset(\mathbb{F}_q^l)^n$ denote $\tilde\pi_i(\mathcal{C})$ by $\mathcal{C}_i,$ for $1\le i\le l.$  Then,
    \begin{enumerate}
        \item[(a)] $\mathcal{C}$ is an $\mathbb{F}_q^l$-submodule of $(\mathbb{F}_q^l)^n$ if and only if $\,\forall\, 1\le j\le l,\, \mathcal{C}_j$ is an $\mathbb{F}_q$-subspace of $\mathbb{F}_q^n$ and $\mathcal{C}=\underset{i=1}{\overset{l}{\sum}} \bm{e}_i\mathcal{C}_i.$
        \item[(b)] If $\mathcal{C}$ is an $\mathbb{F}_q^l$-submodule of $(\mathbb{F}_q^l)^n$ and $\mathcal{C}=\underset{i=1}{\overset{l}{\sum}} \bm{e}_i\mathcal{C}_i'$, where $\mathcal{C}_i'\subset\mathbb{F}_q^n$ ($1\le i\le l$),  then $\mathcal{C}_i'=\mathcal{C}_i$ for each $i$ so that the sum is direct. 
    \end{enumerate}
\end{theorem}
\begin{proof}
\begin{enumerate}
    \item [(a)]Let $\mathcal{C}$ be an $\mathbb{F}_q^l$-submodule of $(\mathbb{F}_q^l)^n$, and let $\Bar{\bm{a}}=(\bm{a}_1,\bm{a}_2,\dots,\bm{a}_n)$    and $\Bar{\bm{b}}=(\bm{b}_1,\bm{b}_2,\dots, \bm{b}_n)\in \mathcal{C}.$ For  $\alpha,\beta\in \mathbb{F}_q,$ observe that  $\alpha\tilde{\pi}_{i}(\Bar{\bm{a}})+\beta\tilde{\pi}_{i}(\Bar{\bm{b}})=\tilde{\pi}_{i}(\alpha\Bar{\bm{a}})+\tilde{\pi}_{i}(\beta\Bar{\bm{b}})=\tilde{\pi}_{i}(\alpha\Bar{\bm{a}}+\beta\Bar{\bm{b}}).$ Since $\alpha\Bar{\bm{a}}+\beta\Bar{\bm{b}}\in \mathcal{C}$,  we have $\alpha\tilde{\pi}_{i}(\Bar{\bm{a}})+\beta\tilde{\pi}_{i}(\Bar{\bm{b}})\in \mathcal{C}_i$ so that $\mathcal{C}_i$ is an $\mathbb{F}_q$-subspace of $\mathbb{F}_q^n.$ Further, by definition of $\mathcal{C}_i$, it follows that $\mathcal{C}=\underset{i=1}{\overset{l}{\sum}}\bm{e}_i\mathcal{C}_i.$  Conversely, suppose that $\mathcal{C}=\underset{i=1}{\overset{l}{\sum}}\bm{e}_i\mathcal{C}_i$ and each $\mathcal{C}_i$ is an $\mathbb{F}_q$-subspace. Let $\Bar{\bm{a}}=\underset{i=1}{\overset{l}{\sum}}\bm{e}_i*{\bm{a}}^{(i)},\,\Bar{\bm{b}}=\underset{i=1}{\overset{l}{\sum}}\bm{e}_i*{\bm{b}}^{(i)}\in\mathcal{C},$ where ${\bm{a}}^{(i)}, {\bm{b}}^{(i)} \in \mathcal{C}_i $ for $1\leq i\leq l,$ and let $\bm{\alpha}=\underset{i=1}{\overset{l}{\sum}}\alpha_i\bm{e}_i\in\mathbb{F}_q^l,$ where $\alpha_i\in \mathbb{F}_q$ for $1\leq i \leq l.$ Then, $\bm{\alpha}\Bar{\bm{a}}+\Bar{\bm{b}}=\underset{i=1}{\overset{l}{\sum}}\bm{e}_i*(\alpha_i\bm{a}^{(i)}+\bm{b}^{(i)})\in \mathcal{C}$ as each $\mathcal{C}_i$ is an $\mathbb{F}_q$-subspace.
    \item[(b)] Since $\{\bm{e}_1,\bm{e}_2,\dots,\bm{e}_l\}$ is an $\mathbb{F}_q$-basis of $\mathbb{F}_q^l,$ the assertion follows. 
\end{enumerate} \hfill $\square$
    
\end{proof}
Thus, $\mathcal{C}=\underset{i=1}{\overset{l}{\bigoplus}} \bm{e}_i\mathcal{C}_i$, where $\mathcal{C}_i=\tilde\pi_i(\mathcal{C}).$ We refer to this decomposition of $\mathcal{C}$ as the {\it $\mathbb{F}_q$-decomposition} with respect to the basis $\{\bm{e}_1,\bm{e}_2,\dots,\bm{e}_l\}$.  
 In part (a) of Theorem \ref{linear}, ``only if" fails if we drop the hypothesis that $\mathcal{C}=\underset{i=1}{\overset{l}{\sum}} \bm{e}_i\mathcal{C}_i.$ 
\begin{example}
   Let $\mathcal{C}=\mathbb{F}_2^2\setminus\{0\}.$ Then,  $\tilde{\pi}_1(\mathcal{C})=\mathbb{F}_2=\tilde{\pi}_2(\mathcal{C}).$ But $\mathcal{C}$ is not $\mathbb{F}_2$-linear.
\end{example}
By Theorem \ref{linear}, given any linear code $\mathcal{C}$ over $\mathbb{F}_q^l$, we obtain $d$ linear codes over $\mathbb{F}_q.$ 

\begin{theorem}\label{polycyclic_classification}
 Let $\mathcal{C}=\underset{i=1}{\overset{l}{\bigoplus}}\bm{e}_i\mathcal{C}_i$ be an $\mathbb{F}_q^l$-linear code of length $n,$ and let $\Bar{\bm{a}}=\bm{e}_1*\bm{a}^{(1)}+\bm{e}_2*\bm{a}^{(2)}+ \dots +\bm{e}_l*\bm{a}^{(l)} \in (\mathbb{F}_q^l)^n$, where $\bm{a}^{(i)}\in\mathbb{F}_q^n$, for $1\le i\le l$. Then, $\mathcal{C}$ is $\Bar{\bm{a}}$-polycyclic over $\mathbb{F}_q^l$ if and only if for each $1 \leq i \leq l$, $\mathcal{C}_i$ is $\bm{a}^{(i)}$-polycyclic over $\mathbb{F}_q$.
\end{theorem}
\begin{proof}Let $\Bar{\bm{c}}=\underset{i=1}{\overset{l}{\sum}}\bm{e}_i*\bm{c}^{(i)}\in\mathcal{C},$ where for each $1\leq i\leq l, \, \, \bm{c}^{(i)}=(c_{1,i},c_{2,i},\dots,c_{n,i} )\in \mathcal{C}_i$. Then, $\Bar{\bm{c}}=(\bm{c}_1,\bm{c}_2,\dots,\bm{c}_{n})$ where $\bm{c}_i=\underset{j=1}{\overset{l}{\sum}}\bm{e}_jc_{i,j}\in\mathbb{F}_q^l.$ Suppose $\Bar{\bm{a}}=(\bm{a}_0,\bm{a}_1,\dots,\bm{a}_{n-1})\in(\mathbb{F}_q^l)^n,$ where $\bm{a}_i=\underset{j=1}{\overset{l}{\sum}}\bm{e}_ja_{i,j},$ for $0\leq i\leq n-1.$ Set, for $1\le i\le l$, $\bm{a}^{(i)}=(a_{0,i},a_{1,i},\dots,a_{n-1,i})\in\mathbb{F}_q^n.$ Observe that $\bm{c}_n\cdot \bm{a}_i=\underset{j=1}{\overset{l}{\sum}}a_{i,j}c_{n,j}\bm{e}_j,$ for $0\leq i\leq n-1,$ and 
   \begin{align*}
     (\bm{0},\bm{c}_1,\dots,\bm{c}_{n-1})+\bm{c}_n\Bar{\bm{a}} = \,\, &  \bm{e}_1*[(0,c_{1,1},c_{2,1},\dots,c_{n-1,1})+c_{n,1}\bm{a}^{(1)}] \\ +\,\, & \bm{e}_2*[(0,c_{1,2},c_{2,2},\dots,c_{n-1,2})+c_{n,2}\bm{a}^{(2)}] \\ + \,\,& \cdots \\ + \,\,& \bm{e}_l*[(0,c_{1,l},c_{2,l},\dots,c_{n-1,l})+c_{n,l}\bm{a}^{(l)}].
   \end{align*}
   Then,  $(0, c_{1,i},c_{2,i},\dots,c_{n-1,i})+c_{n,i}\bm{a}^{(i)}\in \mathcal{C}_i$ for  each $1\le i\le l$ if and only if $(\bm{0},\bm{c}_1,\dots,\bm{c}_{n-1}) +\bm{c}_n\Bar{\bm{a}} \in \mathcal{C}.$ Hence, each $\mathcal{C}_i$ is $\bm{a}^{(i)}$-polycyclic over $\mathbb{F}_q$ if and only if $\mathcal{C}$ is 
   $\Bar{\bm{a}}$-polycyclic over $\mathbb{F}_q^l.$ \hfill $\square$
\end{proof}
%If we choose $\mathbf{a}=(u,0,0,\dots,0)\in S^n$, then Theorem \ref{polycyclic_classification} becomes [Lemma 3, cite Jian Gao]. Hence, Theorem \ref{polycyclic_classification} generalizes [Lemma 3, cite Jian Gao].
\textbf{Note}: From Remark \ref{remark2.10PIR}, it is clear that every $\Bar{\bm{a}}$-polycyclic code over $\mathbb{F}_q^l$ is principally generated. Moreover, a generator polynomial can be given as:
\begin{theorem}\label{Generator polynomial of Polycyclic codes}
    Let $\mathcal{C}$ be an $\Bar{\bm{a}}$-polycyclic code over $\mathbb{F}_q^l$ and let 
 $\underset{i=1}{\overset{l}{\bigoplus}}\bm{e}_i\mathcal{C}_i$ be the $\mathbb{F}_q$-
 decomposition of $\mathcal{C}$ with respect to $\mathcal{B}$. If $\mathcal{C}_i$ is $\mathbb{F}_q$-isomorphic to the ideal generated by $\bm{g}^{(i)}(x)$ in $\mathbb{F}_q[x]/\langle x^n-\bm{a}^{(i)}(x)\rangle$, then $\Bar{\bm{g}}(x):=\underset{i=1}{\overset{l}{
    \sum}} \bm{e}_i\bm{g}^{(i)}(x)$ divides $x^n-\Bar{\bm{a}}(x)$ in $\mathbb{F}_q^l[x],$ and the ideal generated by $\Bar{\bm{g}}(x)$ in $\mathbb{F}_q^l[x]/\langle x^n-\Bar{\bm{a}}(x)\rangle$ is isomorphic to $\mathcal{C}$ as an $\mathbb{F}_q^l$-module. Moreover, if $\mathcal{C}=\langle \Bar{\bm{g}}(x) \rangle$, then, $\mathcal{C}_i=\langle \bm{g}^{(i)}(x) \rangle$ for each $1 \leq i \leq l,$ where $\Bar{\bm{g}}(x)=\underset{i=1}{\overset{l}{\sum}}\bm{g}^{(i)}(x).$
\end{theorem} 
\begin{proof}
    Assume the hypothesis of the theorem. Then, there exists $\bm{h}^{(i)}(x)\in \mathbb{F}_q[x]$ such that $\bm{g}^{(i)}(x)\bm{h}^{(i)}(x)=x^n-\bm{a}^{(i)}(x).$ Define, $\Bar{\bm{h}}(x)=\underset{i=1}{\overset{l}{\sum}}\bm{e}_i\bm{h}^{(i)}(x)\in\mathbb{F}_q^l[x].$ Then, $\Bar{\bm{g}}(x)\Bar{\bm{h}}(x)=\underset{i=1}{\overset{l}{\sum}}\bm{e}_i\bm{g}^{(i)}(x)\bm{h}^{(i)}(x)=\underset{i=1}{\overset{l}{\sum}}\bm{e}_i(x^n-\bm{a}^{(i)}(x))=x^n-\Bar{\bm{a}}(x)$ so that  $\Bar{\bm{g}}(x) \mid (x^n-\Bar{\bm{a}}(x))$ over $\mathbb{F}_q^l.$  Next, let  $\overline{\Bar{\bm{c}}(x)}\in\mathcal{C}.$ Since $\mathcal{C}=\underset{i=1}{\overset{l}{\bigoplus}}\bm{e}_i\mathcal{C}_i, \overline{\Bar{\bm{c}}(x)}=\underset{i=1}{\overset{l}{\sum}}\bm{e}_i \overline{\bm{g}^{(i)}(x) \bm{q}^{(i)}(x)}$ for $\bm{q}^{(i)}(x)\in \mathbb{F}_q[x].$ Observe that $\overline{\Bar{\bm{c}}(x)}=\overline{\Bar{\bm{q}}(x)}\, \overline{\Bar{\bm{g}}(x)}$ for $\Bar{\bm{q}}(x)=\left(\underset{i=1}{\overset{l}{\sum}}\bm{e}_i\bm{q}^{(i)}(x)\right)\in\mathbb{F}_q^l[x].$ The converse follows from the fact that, for each $1 \leq i \leq l$, $\bm{e}_i\mathcal{C}_i$ is embedded in $\mathcal{C}.$  \hfill $\square$
    \end{proof}
    \begin{corollary}
        Let $\Bar{\bm{a}}=(\bm{a}_0,\bm{a}_1,\dots,\bm{a}_{n-1})\in(\mathbb{F}_q^l)^n,$ and let $x^n-\Bar{\bm{a}}(x)=\underset{i=1}{\overset{l}{\sum}}\bm{e}_i(x^n-\bm{a}^{(i)}(x)).$ Suppose $x^n-\bm{a}^{(i)}(x)=f_{i,1}(x)^{n_{i,1}}\dots f_{i,{k_i}}(x)^{n_{i,{k_i}}}$ is the prime factorization of $x^n-\bm{a}^{(i)}(x) $ over $\mathbb{F}_q$ for $1\leq i\leq l.$ Then, there are exactly  $\underset{i=1}{\overset{l}{\prod}}\underset{j=1}{\overset{k_i}{\prod}}(n_{i,j}+1)$ distinct $\Bar{\bm{a}}$-polycyclic codes over $\mathbb{F}_q^l$. 
    \end{corollary}
    \begin{proof}
    Let $\Bar{\bm{a}}(x)\in \mathbb{F}_q^l[x]$ and let $x^n-\Bar{\bm{a}}(x)=\underset{i=1}{\overset{l}{\sum}}\bm{e}_i(x^n-\bm{a}^{(i)}(x)).$  Since every $\Bar{\bm{a}}$-polycyclic code $\mathcal{C}$ is an ideal of $\mathbb{F}_q^l[x]/\langle x^n-\Bar{\bm{a}}(x)\rangle$ and can be uniquely decomposed as $\mathcal{C}=\underset{i=1}{\overset{l}{\bigoplus}}\bm{e}_i\mathcal{C}_i,$  where each $\mathcal{C}_i$ is an $\mathbb{F}_q$-linear subspace of $\mathbb{F}_q^n$ and a principal ideal of $\mathbb{F}_q[x]/\langle x^n-\bm{a}^{(i)}(x) \rangle.$ The number of distinct $\Bar{\bm{a}}$-polycyclic codes over $\mathbb{F}_q^l$ is equal to the product of the  number of distinct $\bm{a}^{(i)}$-polycyclic codes over $\mathbb{F}_q$ for $1\le i\le l.$ The number of distinct $\bm{a}^{(i)}$-polycyclic codes over $\mathbb{F}q$ is equal to the number of divisors of $x^n-\bm{a}^{(i)}(x).$ \hfill $\square$
    \end{proof} 
   
    \begin{corollary}   
    Let $\mathcal{C}$ and $\widetilde{\mathcal{C}}$ be $\Bar{\bm{a}}$-polycyclic codes and $\Bar{\bm{g}}(x), \widetilde{\Bar{\bm{g}}(x)}$ (respectively) are their generator polynomials over $\mathbb{F}_q^l$. Then, $\mathcal{C} \subset \widetilde{\mathcal{C}}$ if and only if $\widetilde{\Bar{\bm{g}}(x)} $ divides $ \Bar{\bm{g}}(x)$ over $\mathbb{F}_q^l$.
\end{corollary}
\begin{proof}
$\mathcal{C} \subset \widetilde{\mathcal{C}} 
\iff
\langle \Bar{\bm{g}}(x) \rangle \subset \langle \widetilde{\Bar{\bm{g}}(x)}\rangle \iff
\Bar{\bm{g}}(x)\in \langle \widetilde{\Bar{\bm{g}}(x)}\rangle
\iff
\Bar{\bm{g}}(x)=\widetilde{\Bar{\bm{g}}(x)}\Bar{\bm{b}}(x)$ for some $\Bar{\bm{b}}(x)\in \mathbb{F}_q^l[x]
\iff
\widetilde{\Bar{\bm{g}}}(x) $ divides $\Bar{\bm{g}}(x)$ over $\mathbb{F}_q^l.$ \hfill $\square$
\end{proof}

 \begin{corollary}
    Let $\mathcal{C}=\langle \Bar{\bm{g}}(x)\rangle$ be an $\Bar{\bm{a}}$-polycyclic code over $\mathbb{F}_q^l$ of length $n.$ If $\Bar{\bm{g}}(x)$ is monic then, $\mathcal{C}$ is a free $\mathbb{F}_q^l$-submodule of $(\mathbb{F}_q^l)^n$ and rank$(\mathcal{C})=n-\deg \Bar{\bm{g}}(x)$.
\end{corollary}
\begin{proof}
    Consider the set $\{\Bar{\bm{g}}(x)\Bar{\bm{h}}(x):\Bar{\bm{h}}(x)\in\mathbb{F}_q^l[x]\}.$ Note that, if $\deg \Bar{\bm{h}}_i(x) \le n-\deg \Bar{\bm{g}}(x)-1$ for $1\le i \le 2$, $\Bar{\bm{g}}(x)\Bar{\bm{h}}_1(x) \not \equiv \Bar{\bm{g}}(x)\Bar{\bm{h}}_2(x) \,  (\textnormal{mod}\, \, x^n-\Bar{\bm{a}}(x)).$ In fact, for any codeword $\Bar{\bm{g}}(x)\Bar{\bm{h}}(x)\in \mathcal{C},$ we write, $\Bar{\bm{g}}(x)\Bar{\bm{h}}(x)=(x^n-\Bar{\bm{a}}(x))\Bar{\bm{q}}(x)+\Bar{\bm{r}}(x),$ where $\deg \Bar{\bm{r}}(x) \le n-1.$ So, we have  $\Bar{\bm{r}}(x)=\Bar{\bm{g}}(x)\Bar{\bm{h}}(x)-(x^n-\Bar{\bm{a}}(x))\Bar{\bm{q}}(x).$ Hence, $\Bar{\bm{g}}(x) \mid \Bar{\bm{r}}(x)$ and $\Bar{\bm{r}}(x)=\Bar{\bm{g}}(x)\Bar{\bm{l}}(x).$ Since $\Bar{\bm{g}}(x)$ is monic, $\deg \Bar{\bm{l}}(x) \le n-\deg \Bar{\bm{g}}(x)-1.$ So, $\{\Bar{\bm{g}}(x),x\Bar{\bm{g}}(x),\dots,x^{n-\deg \Bar{\bm{g}}(x)-1}\Bar{\bm{g}}(x)\}$ spans $\mathcal{C}.$ Also, as $\Bar{\bm{g}}(0)$ is unit, $\{\Bar{\bm{g}}(x),x\Bar{\bm{g}}(x),\dots,x^{n-\deg \Bar{\bm{g}}(x)-1}\Bar{\bm{g}}(x)\}$ is linearly independent over $\mathbb{F}_q^l.$ \hfill $\square$
\end{proof}
\begin{corollary}
    Let $\mathcal{C}=\langle \Bar{\bm{g}}(x) \rangle$ be an $\Bar{\bm{a}}$-polycyclic code  over $\mathbb{F}_q^l$. Then $\Bar{\bm{g}}(x)$ is monic if and only if  $\deg(\bm{g}^{(i)}(x))=\deg(\bm{g}^{(j)}(x))$ for $1\le i, j\le l,$ where $\Bar{\bm{g}}(x)=\underset{i=1}{\overset{l}{\sum}}\bm{g}^{(i)}(x)$.
\end{corollary}
\begin{proof}
    If each $\bm{g}^{(i)}(x)$ has same degree, then $\underset{i=1}{\overset{l}{\sum}}\bm{e}_i=1\implies \Bar{\bm{g}}(x)$ is monic. Conversely, sum of $l-1$ or less elements in $\{\bm{e}_1,\dots,\bm{e}_l\}$ is a zero divisor in $\mathbb{F}_q^l$. \hfill $\square$
\end{proof}

In 2016, A. Alahmadi (\cite{alahmadi2016duality}) introduced following bilinear form.
\begin{definition}\label{bilinearform} Let $\mathbf{f}(x) \in \mathbb{F}_q[x]$.
   Define the $\mathbb{F}_q$-valued bilinear form on $\mathbb{F}_q[x]/\langle \mathbf{f}(x)\rangle$ by
    \begin{equation}
    \langle \overline{\mathbf{h}_1(x)},\overline{\mathbf{h}_2(x)} \rangle_\mathbf{f}=\mathbf{r}(0)    
    \end{equation}
    for $\mathbf{h}_1(x), \mathbf{h}_2(x)\in \mathbb{F}_q[x],$ where $\mathbf{r}(x)$ is the remainder when $\mathbf{h}_1(x)\mathbf{h}_2(x)$ is divided by $\mathbf{f}(x)$ in $\mathbb{F}_q[x].$ 
\end{definition}
This bilinear form can be extended and defined over $\frac{\mathbb{F}_q^l[x]}{\langle x^n-\Bar{\bm{a}}(x) \rangle}$ in the following manner.
Define the $\mathbb{F}_q^l$-valued bilinear form on the quotient $\frac{\mathbb{F}_q^l[x]}{\langle x^n-\Bar{\bm{a}}(x) \rangle }$ by
    \begin{equation}
    \langle \overline{\Bar{\bm{h}}_1(x)},\overline{\Bar{\bm{h}}_2(x)} \rangle_{\Bar{\bm{a}}}=\Bar{\bm{r}}(0)    
    \end{equation}
    for $\Bar{\bm{h}}_1(x), \Bar{\bm{h}}_2(x)\in \mathbb{F}_q^l[x],$ where $\Bar{\bm{r}}(x)$ is the remainder when $\Bar{\bm{h}}_1(x)\Bar{\bm{h}}_2(x) $ is divided by $x^n-\Bar{\bm{a}}(x)$ in $ \mathbb{F}_q^l[x].$ 

\begin{lemma} If $\Bar{\bm{a}}(0)=\bm{a}_0$ is a unit in $\mathbb{F}_q^l$, then the  bilinear form in Definition \ref{bilinearform} is non-degenerate. 
\end{lemma}
\begin{proof}Let $\Bar{\bm{h}}(x)\in\mathbb{F}_q^l[x]$ be such that $\langle \overline{\Bar{\bm{f}}(x)}, \overline{\Bar{\bm{h}}(x)} \rangle_{\Bar{\bm{a}}}=0\,  \forall \, \Bar{\bm{f}}(x)\in \mathbb{F}_q^l[x].$ Assume, without loss of generality, $\deg(\Bar{\bm{h}}(x))\le n-1;$ and set $\Bar{\bm{h}}(x)=\bm{h}_0+\bm{h}_1x+\dots+\bm{h}_{n-1}x^{n-1}$.  Putting $\Bar{\bm{f}}(x)=1$, we get $\bm{h}_0=\Bar{\bm{h}}(0)=\bm{0}.$ Inductively, we get  for $1\le i\le n-1,$ putting $\Bar{\bm{f}}(x)=x^{n-i}$, that $\bm{a}_0\bm{h}_{i}=\bm{0}$ and hence $\bm{h}_{i}=\bm{0}$ as $\bm{a}_0$ is a unit so that $\Bar{\bm{h}}(x)=\Bar{\bm{0}}.$ \hfill $\square$
\end{proof}
However, the converse is false.
\begin{example}
    For $\Bar{\bm{a}}(x)=\bm{0}$ and $n=1$, in $\mathbb{F}_q^l[x]/\langle x \rangle$ the above annihilator product reduces to the product in the ring $\mathbb{F}_q^l,$ and hence it becomes a non-degenerate bilinear form. 
\end{example}
\begin{definition} The bilinear form below the Definition \ref{bilinearform} is called {\it annihilator product.} 
    We denote the dual of $\mathcal{C}$ with respect to the annihilator product by $\mathcal{C}^\circ,$ and call it the {\it annihilator dual.} Symbolically, $\mathcal{C}^\circ:=\{\overline{\Bar{\bm{f}}(x)}\in\frac{\mathbb{F}_q^l[x]}{\langle x^n-\Bar{\bm{a}}(x)\rangle}|\,\langle \overline{\Bar{\bm{f}}(x)},\overline{\Bar{\bm{h}}(x)} \rangle_{\Bar{\bm{a}}}=\bm{0}, \forall\, \overline{\Bar{\bm{h}}(x)} \in \mathcal{C}\}.$
\end{definition}
\begin{definition}
    The  {\it annihilator} of $\mathcal{C}$, denoted by  
    \begin{equation}
    \textnormal{ Ann}(\mathcal{C}):=\{\overline{\Bar{\bm{f}}(x)}\in\frac{\mathbb{F}_q^l[x]}{\langle x^n-\Bar{\bm{a}}(x)\rangle}|\, \overline{\Bar{\bm{f}}(x)}\,\overline{\Bar{\bm{h}}(x)}=\Bar{\bm{0}} \in \frac{\mathbb{F}_q^l[x]} { \langle x^n-\Bar{\bm{a}}(x)\rangle} \textnormal{ for all } \overline{\Bar{\bm{h}}(x)} \in \mathcal{C}\}.
    \end{equation}
\end{definition}

\begin{lemma}\label{Ann(C)=C^0}
     Let $\Bar{\bm{a}}(x)=\bm{a}_{n-1}x^{n-1}+\dots+\bm{a}_1x+\bm{a}_0 \in \mathbb{F}_q^l[x]$, where $\bm{a}_0 \in U(\mathbb{F}_q^l)$ and let $\mathcal{C}$ be an $\Bar{\bm{a}}$-polycyclic code over $\mathbb{F}_q^l$ having length $n$. Then $\mathcal{C}^\circ=\textnormal{Ann}(\mathcal{C}).$
\end{lemma}
\begin{proof}By definition, $\textnormal{Ann}(\mathcal{C})\subset\mathcal{C}^\circ.$ For the converse, let $\overline{\Bar{\bm{f}}(x)}\in\mathcal{C}^\circ.$ We require to show that $\overline{\Bar{\bm{f}}(x)\Bar{\bm{h}}(x)}=\Bar{\bm{0}}$ for every $\overline{\Bar{\bm{h}}(x)}\in\mathcal{C}.$  Suppose that $\Bar{\bm{g}}(x)$ generates $\mathcal{C}$,  $\Bar{\bm{f}}(x)\Bar{\bm{g}}(x)=(x^n-\Bar{\bm{a}}(x))\Bar{\bm{q}}(x)+\Bar{\bm{r}}(x)$ where $\Bar{\bm{q}}(x), \Bar{\bm{r}}(x)\in \mathbb{F}_q^l[x]$ and $\Bar{\bm{r}}(x)=\bm{r}_{n-1}x^{n-1}+\dots+\bm{r}_1x+\bm{r}_0.$ 
Then, for each $i\ge0,$ $\overline{x^i\Bar{\bm{g}}(x)}\in\mathcal{C}$ so that $\bm{r}_0=\langle \overline{\Bar{\bm{f}}(x)}, \overline{\Bar{\bm{g}}(x)} \rangle_{\Bar{\bm{a}}}=\bm{0},\, \bm{a}_0\bm{r}_{n-1}=\langle \overline{\Bar{\bm{f}}(x)}, \overline{x\Bar{\bm{g}}(x)} \rangle_{\Bar{\bm{a}}}=\bm{0},\dots, \bm{a}_0\bm{r}_1=\langle \overline{\Bar{\bm{f}}(x)}, \overline{x^{n-1}\Bar{\bm{g}}(x)} \rangle_{\Bar{\bm{a}}}=\bm{0}.$  Since $\bm{a}_0$ is a unit, $\bm{r}_j=0$ for each $j$ so that $\overline{\Bar{\bm{f}}(x)\Bar{\bm{g}}(x)}=\Bar{\bm{0}}.$   \hfill $\square$
 \end{proof}
We have a few results which hold for sequential codes over $\mathbb{F}_q^l.$
\begin{theorem}\label{classification of sequential codes}Let $\mathcal{C}$ be an $\mathbb{F}_q^l$-linear code of length $n,$ and let $\Bar{\bm{a}}=\underset{i=1}{\overset{l}{\sum}} \bm{e}_i*\bm{a}^{(i)}\in (\mathbb{F}_q^l)^n,$ where $\bm{a}^{(i)}\in\mathbb{F}_q^n,$ for $1\le i\le l$. Then, $\mathcal{C}$ is $\Bar{\bm{a}}$-sequential over $\mathbb{F}_q^l$ if and only if $\mathcal{C}_i$ is $\bm{a}^{(i)}$-sequential over $\mathbb{F}_q$ for each $i.$
    \end{theorem}
    \begin{proof}
        The proof of the above theorem follows on the same lines as of the Theorem \ref{polycyclic_classification}.
    \end{proof}
%     \begin{proof}
% Let $\mathbf{c}=\underset{i=1}{\overset{l}{\sum}}\bm{e}_i\mathbf{c}_i\in\mathcal{C},$ where $\mathbf{c}_i=(c_{1,i},c_{2,i},\dots,c_{n,i} )\in \mathcal{C}_i$ for every $1\leq i\leq l$. Then, $\mathbf{c}=(d_1,d_2,\dots,d_{n})$ where $d_i=\underset{j=1}{\overset{l}{\sum}}c_{i,j}e_j\in\mathbb{F}_q^l.$ Suppose $\mathbf{a}=(a_0,a_1,\dots,a_{n-1})\in\mathbb{F}_q^{nl},$ where $a_i=\underset{j=1}{\overset{l}{\sum}}b_{i,j}e_j.$ Set, for $1\le i\le l$, $\mathbf{a}_i=(b_{0,i},b_{1,i},\dots,b_{n-1,i})\in\mathbb{F}_q^n.$ Observe that $d_{i}\cdot a_{i-1}=\underset{j=1}{\overset{l}{\sum}}c_{i,j}\, b_{i-1,j}\, e_j,$ and 
%   \begin{align*}
%      (d_2,d_3,\dots,\mathbf{c}\cdot\mathbf{a})  = & \,\bm{e}_1(c_{2,1},\dots,c_{n,1},\mathbf{c}_1\cdot \mathbf{a}_1)  + \bm{e}_2(c_{2,2},\dots,c_{n,2},\mathbf{c}_2\cdot \mathbf{a}_2)  + \, \dots\\& + \,\bm{e}_l(c_{2,l},\dots,c_{n,l},\mathbf{c}_l\cdot \mathbf{a}_l).
%    \end{align*}
%    Thus,  $(c_{2,i},\dots,c_{n,i},\mathbf{c}_i\cdot \mathbf{a}_i)\in \mathcal{C}_i$ for  each $1\le i\le l$ if and only if $(d_2,d_3,\dots,d_n,\mathbf{c}\cdot\mathbf{a}) \in \mathcal{C}.$ Hence, each $\mathcal{C}_i$ is $\mathbf{a}_i$-sequential over $\mathbb{F}_q$ if and only if $\mathcal{C}$ is 
%    $\mathbf{a}$-sequential over $\mathbb{F}_q^l.$  \hfill $\square$      
%     \end{proof}
  %  \textbf{Notation:}
        %$\mathcal{C}^\perp$ denotes the Euclidean Dual of $\mathcal{C}.$
 
    \begin{lemma} An $\mathbb{F}_q^l$-linear code is $\Bar{\bm{a}}$-polycyclic if and only if its Euclidean dual is $\Bar{\bm{a}}$-sequential over $\mathbb{F}_q^l.$
           \end{lemma}
    \begin{proof}By Theorem \ref{polycyclic_classification}, $\mathcal{C}$ is $\Bar{\bm{a}}$-polycyclic code over $\mathbb{F}_q^l$ if and only if each $\mathcal{C}_i$ is $\bm{a}^{(i)}$-polycyclic over $\mathbb{F}_q.$ By [\cite{lopez2009dual},Theorem 3.2], $\mathcal{C}_i$ is $\bm{a}^{(i)}$-polycyclic over $\mathbb{F}_q$ if and only if $\mathcal{C}_i^\perp$ is $\bm{a}^{(i)}$-sequential over $\mathbb{F}_q.$ By Theorem \ref{classification of sequential codes} $\mathcal{C}^\perp$ is $\Bar{\bm{a}}$-sequential over $\mathbb{F}_q^l.$ \hfill $\square$
    \end{proof}
    \begin{lemma}\label{Dual of dual}
       If  $\mathcal{C}$ is an $\mathbb{F}_q^l$-linear $\Bar{\bm{a}}$-polycyclic code, then  $\mathcal{C}^\circ=(\mathcal{C}. A)^\perp$ where $\mathcal{C}. A=\{\Bar{\bm{c}}A\,|\,\Bar{\bm{c}}\in\mathcal{C}\}$ with $A:=(\langle \Bar{\bm{\epsilon}}_i(x),\Bar{\bm{\epsilon}}_j(x)\rangle_{\Bar{\bm{a}}})_{1\leq i,j\leq n}$, for the standard basis $\{\Bar{\bm{\epsilon}}_i \,|\,1\le i\le n\}$ of $(\mathbb{F}_q^l)^n$ over $\mathbb{F}_q^l.$ In particular,  $(\mathcal{C}^\circ)^\circ=\mathcal{C}.$
    \end{lemma}
    \begin{proof}
        Note that, $\langle \Bar{\bm{x}}, \Bar{\bm{y}}\rangle_{\Bar{\bm{a}}}=\Bar{\bm{x}}A\Bar{\bm{y}}^t$ for $\Bar{\bm{x}}, \Bar{\bm{y}}\in (\mathbb{F}_q^l)^n.$ But $A^t=A$ and hence, $\Bar{\bm{x}}A\Bar{\bm{y}}^t=\Bar{\bm{x}}(\Bar{\bm{y}}A)^t=\langle \Bar{\bm{x}},\Bar{\bm{y}}A\rangle.$ Further, consider $\mathcal{C}^\circ=\{\Bar{\bm{x}}\in(\mathbb{F}_q^l)^n | \langle \Bar{\bm{x}}, \Bar{\bm{y}}\rangle_{\Bar{\bm{a}}} =\bm{0}, \, \forall \,\Bar{\bm{y}}\in \mathcal{C}\} =\{\Bar{\bm{x}}\in(\mathbb{F}_q^l)^n | \langle \Bar{\bm{x}},\Bar{\bm{y}}A\rangle=\bm{0}, \, \forall \,\Bar{\bm{y}}\in \mathcal{C}\}=\{\Bar{\bm{x}}\in(\mathbb{F}_q^l)^n| \langle \Bar{\bm{x}},\Bar{\bm{y}}\rangle=\bm{0}, \, \forall \,\Bar{\bm{y}}\in \mathcal{C}. A\}=(\mathcal{C}. A)^\perp.$\\
        Observe that, $\Bar{\bm{y}}\in (\mathcal{C}.A)^\perp \iff 
        \langle \Bar{\bm{y}}, \Bar{\bm{x}}A \rangle =\bm{0}\,\, \forall \,\Bar{\bm{x}}\in \mathcal{C} 
        \iff 
        \Bar{\bm{y}}(\Bar{\bm{x}}A)^\perp=\bm{0}\,\, \forall\, \Bar{\bm{x}}\in \mathcal{C}
        \iff
        \langle \Bar{\bm{y}}A, \Bar{\bm{x}} \rangle =\bm{0} \, \forall \, \Bar{\bm{x}}\in \mathcal{C}
        \iff 
        \Bar{\bm{y}}A\in \mathcal{C}^\perp
        \iff
        \Bar{\bm{y}}\in \mathcal{C}^\perp. A^{-1}.$ Hence, $(\mathcal{C}. A)^\perp=\mathcal{C}^\perp .A^{-1}.$
        Therefore, $(\mathcal{C}^\circ)^\circ=((\mathcal{C}. A)^\perp)^\circ=(\mathcal{C}^\perp. A^{-1})^\circ=\mathcal{C}.$ \hfill $\square$
       % Follows from the arguments of the proof of  \cite{tabue2018polycyclic} lemma 3.7. 
    \end{proof}
    \begin{notation}
         We will denote $A$ in Lemma \ref{Dual of dual} as $(\langle \Bar{\bm{\epsilon}}_i,\Bar{\bm{\epsilon}}_j\rangle_{\Bar{\bm{a}}})$ for convinience. 
    \end{notation}
    \begin{theorem}\label{Dual decomposition}
        For $\Bar{\bm{a}}$-polycyclic code $\mathcal{C}$ over $\mathbb{F}_q^l,$ if the $\mathbb{F}_q$-decomposition with respect to $\mathcal{B}$ is  $\mathcal{C}=\bigoplus_{i=1}^l{\bm{e}_i\,\mathcal{C}_i}$, then the $\mathbb{F}_q$-decomposition of $\mathcal{C}^\circ$ with respect to $\mathcal{B}$ is given by $\mathcal{C}^\circ=\bigoplus_{i=1}^l{\bm{e}_i\,\mathcal{C}^\circ_i}.$
    \end{theorem}
    \begin{proof}Suppose $A$ is as in Lemma \ref{Dual of dual}. Then $A=\underset{i=1}{\overset{l}{\sum}}\bm{e}_iA_i,$ where $A_i\in M_n(\mathbb{F}_q).$ By Lemma \ref{Dual of dual},
       $$\mathcal{C}^\circ=(\mathcal{C}.A)^\perp
            =\underset{i=1}{\overset{l}{\bigoplus}}\bm{e}_i\left(\mathcal{C}_i. A_i\right)^\perp
            =\underset{i=1}{\overset{l}{\bigoplus}}\bm{e}_i\mathcal{C}_i^\circ.$$ \hfill $\square$
    \end{proof}
    \begin{theorem}
        Let $\Bar{\bm{a}}\in(\mathbb{F}_q^l)^n$. Then, an $\mathbb{F}_q^l$-linear code is $\Bar{\bm{a}}$-polycyclic if and only if its annihilator dual is $\Bar{\bm{a}}$-polycyclic.
    \end{theorem}
    \begin{proof}
    Let $\mathcal{C}$ be $\Bar{\bm{a}}$-polycyclic code over $\mathbb{F}_q^l,$ and let $\mathcal{C}=\underset{i=1}{\overset{l}{\bigoplus}}\bm{e}_i\mathcal{C}_i.$ By Theorem \ref{Dual decomposition}, $\mathcal{C}^\circ=\underset{i=1}{\overset{l}{\bigoplus}}\bm{e}_i\mathcal{C}_i^\circ.$ By Theorem \ref{polycyclic_classification}, each $\mathcal{C}_i$ is $\bm{a}^{(i)}$-polycyclic. By [Proposition 3, \cite{tabue2018polycyclic}], it follows that each $\mathcal{C}_i^\circ$ is $\bm{a}^{(i)}$-polycyclic as $\mathbb{F}_q$ is a chain ring. Again, by Theorem \ref{polycyclic_classification}, $\mathcal{C}^\circ$ is $\mathbf{a}$-polycyclic. Converse follows from Lemma \ref{Dual of dual}.\hfill $\square$
    \end{proof}
    \begin{corollary}\label{generatorpolynomialofC^0}
        Let $\mathcal{C}=\langle \Bar{\bm{g}}(x) \rangle$ be an $\Bar{\bm{a}}$-polycyclic code over $\mathbb{F}_q^l,$ and let $\Bar{\bm{h}}(x)$ is a check polynomial of $
        \mathcal{C}$ so that $x^n-\Bar{\bm{a}}(x)=\Bar{\bm{g}}(x)\Bar{\bm{h}}(x)$. Then, $\mathcal{C}^\circ$ is an $\Bar{\bm{a}}$-polycyclic code and is generated by $\Bar{\bm{h}}(x).$
    \end{corollary}
    \begin{proof} Suppose
    $\Bar{\bm{g}}(x)=\underset{i=1}{\overset{l}{
    \sum}} \bm{e}_i\bm{g}^{(i)}(x)$ and $\Bar{\bm{h}}(x)=\underset{i=1}{\overset{l}{
    \sum}} \bm{e}_i\bm{h}^{(i)}(x).$
        By Theorem \ref{Generator polynomial of Polycyclic codes}, $\mathcal{C}_i=\langle \bm{g}^{(i)}(x) \rangle,$  for each $1\le i\le l.$ Then by [Lemma 2.4 in \cite{qi2022polycyclic}], $\mathcal{C}^\circ_i=\langle \bm{h}^{(i)}(x) \rangle,$ for each $1\le i\le l.$ So, by Theorem \ref{Dual decomposition}, the result follows. \hfill $\square$
    \end{proof}
    \begin{corollary}
        Let $\mathcal{C}=\langle \Bar{\bm{g}}(x)\rangle$ be an $\Bar{\bm{a}}$-polycyclic code over $\mathbb{F}_q^l,$ and let $\Bar{\bm{g}}(x)$ be a check polynomial of $\mathcal{C}.$ Then $\mathcal{C}$ is annihilator self orthogonal if and only if $\Bar{\bm{h}}(x) $ divides $ \Bar{\bm{g}}(x)$ in $\mathbb{F}_q^l[x].$ 
    \end{corollary} 
\begin{proof}
    By Corollary \ref{generatorpolynomialofC^0}, we have $\mathcal{C}^\circ=\langle \Bar{\bm{h}}(x)\rangle.$ Hence, $\mathcal{C}\subset\mathcal{C}^\circ \iff \langle \Bar{\bm{g}}(x) \rangle \subset \langle \Bar{\bm{h}}(x) \rangle \iff \Bar{\bm{g}}(x)\in \langle \Bar{\bm{h}}(x) \rangle \iff \Bar{\bm{h}}(x) \mid \Bar{\bm{g}}(x). $ \hfill $\square$
\end{proof}
   \begin{definition}
  Let $\mathcal{C}$ be an $\mathbb{F}_q^l$-linear code,
   \begin{itemize}
   \item It is called {\it annihilator self-orthogonal} if $\mathcal{C} \subset \mathcal{C}^\circ,$
   \item It is called {\it annihilator self-dual } if $\mathcal{C}=\mathcal{C}^\circ,$
   \item It is called {\it annihilator dual-containing} if $\mathcal{C}^\circ \subset  \mathcal{C} ,$
  \item  It is called an {\it annihilator LCD code} if $\mathcal{C} \cap \mathcal{C}^\circ=\{0\}.$  
   \end{itemize}

   \end{definition}
   From Theorem \ref{Dual decomposition} the following Corollary is immediate:
    \begin{corollary}\label{CiffC_i}
        If $\mathcal{C}$ be an $\mathbf{a}$-polycyclic code over $\mathbb{F}_q^l.$ Then,
        \begin{enumerate}
        \item $\mathcal{C}$ is annihilator self-orthogonal $\Leftrightarrow$  each $\mathcal{C}_i$ is self-orthogonal over $\mathbb{F}_q,$
            \item $\mathcal{C}$ is an annihilator dual-containing $\Leftrightarrow$ each  $\mathcal{C}_i$ is dual-containing over $\mathbb{F}_q,$
            \item $\mathcal{C}$ is annihilator self-dual $\Leftrightarrow$  each $\mathcal{C}_i$ is self-dual over $\mathbb{F}_q,$
            \item C is annihilator LCD $\Leftrightarrow$ each $\mathcal{C}_i$ is annihilator LCD over $\mathbb{F}_q.$
        \end{enumerate}
    \end{corollary} \hfill $\square$\\
The following result is due to [Corollary 3, \cite{bajalan2024polycyclic}]:
    \begin{theorem}\label{Ann dual containing over F_q}
    Let $\mathcal{C}=\langle \Bar{\bm{g}}(x) \rangle $ be a $\Bar{\bm{a}}$-polycyclic code of length $n$ over $\mathbb{F}_q.$  Suppose $x^n-\Bar{\bm{a}}(x)=\Bar{\bm{h}}(x)\Bar{\bm{g}}(x).$ Then $\mathcal{C}^\circ \subseteq \mathcal{C}$ if and only if $\Bar{\bm{g}}(x)$ is a divisor of $\Bar{\bm{h}}(x).$
\end{theorem}
% \begin{proof}
%     If $\mathcal{C}^\circ\subseteq \mathcal{C}$, then $\bm{h}(x)\in\mathcal{C}^\circ\subseteq\mathcal{C}.$ Therefore, $\bm{h}(x) = \bm{g}(x) \bm{q}(x)$ for some $\bm{q}(x) \in \mathbb{F}_{q}[x]$. Hence, $\bm{h}(x)\bm{h}(x) = \bm{h}(x)\bm{g}(x)\bm{q}(x)=(x^n - \mathbf{a}(x))\bm{q}(x).$ This shows that $x^n - \mathbf{a}(x)$ divides $\bm{h}(x)^2.$ Conversely, assume that $x^n - \mathbf{a}(x)$ divides $\bm{h}(x)^2.$ Write $\bm{h}(x)^2=(x^n-\mathbf{a}(x))\bm{r}(x)$ for $\bm{r}(x)\in \mathbb{F}_q[x].$ Now, if $\bm{b}(x)\in \mathcal{C}^\circ=\langle \bm{h}(x) \rangle,$ then $\bm{b}(x)=\bm{h}(x)\bm{q}(x) $ for some $\bm{q}(x)\in\mathbb{F}_q[x].$ Hence, $\bm{b}(x)\bm{h}(x)=\bm{q}(x)\bm{h}(x)\bm{h}(x)=\bm{q}(x)(x^n-\mathbf{a}(x))\bm{r}(x)=\bm{q}(x)\bm{g}(x)\bm{h}(x)\bm{r}(x)$ and consequently, $\bm{b}(x)\in \mathcal{C}. $
% \end{proof}
    
     Now from Theorem 2.7, Theorem 2.8, Theorem 3.3, Theorem 4.6  of \cite{qi2022polycyclic} and from Theorem \ref{Ann dual containing over F_q}, Corollary \ref{CiffC_i} we get the following results:
     \begin{corollary} \label{decomposition}
     Let $\mathcal{C}$ be an $\Bar{\bm{a}}$-polycyclic code over $\mathbb{F}_q^l.$ Then,         \begin{enumerate}
         \item $\mathcal{C}=\langle \Bar{\bm{g}}(x) \rangle=\langle \underset{i=1}{\overset{l}{\sum}}\bm{e}_i\bm{g}^{(i)}(x) \rangle$ is annihilator dual-containing over $\mathbb{F}_q^l \Leftrightarrow $ $\bm{g}^{(i)}(x)\mid\bm{h}^{(i)}(x)$ for each $1\leq i\leq l.$  
         \item $\mathcal{C}=\langle \Bar{\bm{g}}(x) \rangle=\langle \underset{i=1}{\overset{l}{\sum}}\bm{e}_i\bm{g}^{(i)}(x) \rangle$ is annihilator self-dual over $\mathbb{F}_q^l \Leftrightarrow x^n-\bm{a}^{(i)}(x)=a_i(\bm{g}^{(i)}(x))^2$ where $ a_i\in U(\mathbb{F}_q)$ for each $1\leq i\leq l.$  
         \item $\mathcal{C}=\langle \Bar{\bm{g}}(x) \rangle=\langle \underset{i=1}{\overset{l}{\sum}}\bm{e}_i\bm{g}^{(i)}(x) \rangle$ is annihilator LCD over $\mathbb{F}_q^l \Leftrightarrow $ $\gcd(\bm{g}^{(i)}(x),\bm{h}^{(i)}(x))=1$ where $x^n-\bm{a}^{(i)}(x)=\bm{g}^{(i)}(x)\bm{h}^{(i)}(x)$ for each $1\leq i\leq l.$   
     \end{enumerate}
             
     \end{corollary}
    \begin{corollary}
        Let $\Bar{\bm{a}}(x)\in\mathbb{F}_q^l[x]$ where $\bm{a}_0 \in U(\mathbb{F}_q^l).$ Assume, for each $1\leq i \leq l,$ that $x^n-\bm{a}^{(i)}(x)=\underset{j=1}{\overset{s_i}{\prod}}f_{i,j}^{m_{i,j}}(x)$  be the factorization into distinct irreducible polynomials $f_{i, j}(x)\in \mathbb{F}_q[x].$ Then,
        \begin{enumerate}
            \item The number of annihilator self-orthogonal $\Bar{\bm{a}}$-polycyclic codes over $\mathbb{F}_q^l$ is $$\underset{i=1}{\overset{l}{\prod}}\underset{j=1}{\overset{s_i}{\prod}}\Big(m_{i, j}-\Big\lceil \frac{m_{i, j}}{2} \Big\rceil +1\Big).$$
       
        \item The number of annihilator self-dual $\Bar{\bm{a}}$-polycyclic  codes over $\mathbb{F}_q^l$ is one if $m_{i, j}$ is even for each $(i, j)$ and zero otherwise.
\item The number of $\Bar{\bm{a}}$-polycyclic annihilator LCD codes over $\mathbb{F}_q^l$
is $2^{\underset{i=1}{\overset{l}{\sum}} s_i}.$     \end{enumerate}
\end{corollary}
\begin{proof}
 From Corollary \ref{CiffC_i}, we have the number of annihilator self-orthogonal $\Bar{\bm{a}}$-polycyclic codes over $\mathbb{F}_q^l$ is the product of number of annihilator self-orthogonal $\bm{a}^{(i)}$-polycyclic codes over $\mathbb{F}_q$ for $1\leq i \leq l.$ The number of annihilator self-orthogonal $\bm{a}^{(i)}$-polycyclic codes over $\mathbb{F}_q$ for $1\leq i \leq l$ can be computed by [Theorem 2.9, \cite{qi2022polycyclic}]. Similarly, the proofs of $2$ and $3$ follow from [Theorem 2.9, \cite{qi2022polycyclic}]. \hfill $\square$
 \end{proof} 
\begin{remark}
    Note that the results proved above for the product $\mathbb{F}_q^l$ also hold for any $\mathbb{F}_q$-algebra isomorphic to the product ring due to Lemma \ref{Existence of basis consisting of orthogonal idempotents}. But product rings enable us to use the readily available standard base for implementation in MAGMA and sagemath.
\end{remark}

\section{Gray Maps}\label{sec5}
% A Gray map for an $\mathbb{F}_q$-algebra $\mathcal{R}$ is an $\mathbb{F}_q$-linear injection from $\mathcal{R}$ to $\mathbb{F}_q^t$ for some $t\ge \dim_{\mathbb{F}_q}\mathcal{R}.$ 
There are several Gray maps that we can define to get codes over $\mathbb{F}_q$ from codes over $\mathbb{F}_q^l.$ Suppose $\{\bm{e}_1,\bm{e}_2,\dots,\bm{e}_l\}$ is an $\mathbb{F}_q $-basis of $\mathbb{F}_q^l$ consisting of orthogonal idempotents having $\underset{i=1}{\overset{l}{\sum}}\bm{e}_i=\bm{1}$ and define a map $\Phi:{(\mathbb{F}_q^l})^n\rightarrow\mathbb{F}_q^{nl}$ such that,
$$\Phi(\bm{c}_1,\dots,\bm{c}_n)=(c_{1,1},c_{1,2},\dots,c_{1,l},\dots,c_{n,1},c_{n,2},\dots,c_{n,l}),$$ where $c_i=\underset{j=1}{\overset{l}{\sum}}e_jc_{i,j}.$ Note that $\Phi$ is a linear transformation from $(\mathbb{F}_q^{l})^n$ to $\mathbb{F}_q^{nl}.$
Let $\Bar{\bm{b}} \in  (\mathbb{F}_q^{l})^n,$ Gray weight of $\Bar{\bm{b}}$ is defined as $w_G(\Bar{\bm{b}})=w_H(\Phi(\Bar{\bm{b}})),$ where $w_H(\Phi(\Bar{\bm{b}}))$ is the Hamming weight of $\Phi(\Bar{\bm{b}})$. Then distance $d_G(\mathcal{C}):=min\{d_G(\Bar{\bm{b}},\Bar{\bm{c}})|\Bar{\bm{b}}\neq\Bar{\bm{c}}\in \mathcal{C}\}$ is clearly the minimum Gray weight of a non-zero codeword.
\begin{theorem}
    Let $\mathcal{C}$ be a linear code over $\mathbb{F}_q^l$ of length $n$ and $\mathcal{C}=\underset{i=1}{\overset{l}{\bigoplus}}\mathcal{C}_i$ be the $\mathbb{F}_q$-decomposition with respect $\mathcal{B}$. Then, $\Phi$ is a distance preserving map from $(\mathbb{F}_q^{l})^n$ to $\mathbb{F}_q^{nl}$. Moreover, $\Phi(\mathcal{C})$ would be an $\mathbb{F}_q$-linear code having parameters $[nl,nl-\underset{i=1}{\overset{l}{\sum}}\deg g_i].$
\end{theorem}
\begin{proof}
    The proof follows directly from the definition of Gray map.   \hfill $\square$
\end{proof} 
% Let $\mathbf{a}=(a_1,a_2,\dots,a_{n})=\bm{e}_1\mathbf{a}_1+\bm{e}_2\mathbf{a}_2+ \dots +\bm{e}_l\mathbf{a}_l \in \mathbb{F}_q^{nl}$ and $a_1$ is unit in $\mathbb{F}_q^l.$ For $r=(r_1,r_2,\dots,r_{n})\in \mathbb{F}_q^{nl},$ and $r_i=\underset{j=1}{\overset{l}{\sum}}e_jr_{i,j}$. Consider, $\mathbf{r}_i=(r_{1,i},r_{2,i},\dots,r_{n,i})$ for $1 \leq i \leq l.$ Then, 
% \begin{equation*}
%     \Phi(r)=(r_{1,1},r_{1,2},\dots,r_{1,l},r_{2,1},r_{2,2},\dots,r_{2,l},\dots,r_{n,1},\dots,r_{n,l}).
% \end{equation*}
\begin{definition}[Polycyclic shift]
    Let $\Bar{\bm{a}}=(\bm{a}_1,\dots,\bm{a}_n)\in(\mathbb{F}_q^l)^n$ and $\bm{a}_1$ is a unit in $\mathbb{F}_q^l.$ Then for $\Bar{\bm{c}}=(\bm{c}_1,\dots,\bm{c}_n)\in (\mathbb{F}_q^l)^n,$ polycyclic shift $\tau_{\Bar{\bm{a}}}$ is defined as $ \tau_{\Bar{\bm{a}}}(\Bar{\bm{c}})=(\bm{0},\bm{c}_1,\dots,\bm{c}_{n-1})+\bm{c}_n(\bm{a}_1,\dots,\bm{a}_n).$
\end{definition}
\begin{definition}[Sequential shift]Let $\Bar{\bm{a}}=(\bm{a}_1,\dots,\bm{a}_n)\in(\mathbb{F}_q^l)^n$ and $\bm{a}_1$ is a unit in $\mathbb{F}_q^l.$ Then for $\Bar{\bm{c}}=(\bm{c}_1,\dots,\bm{c}_n)\in (\mathbb{F}_q^l)^n,$ sequential shift $\widetilde{\tau}_{\Bar{\bm{a}}}$ is defined as
$ \widetilde{\tau}_{\Bar{\bm{a}}}(\Bar{\bm{c}})=(\bm{c}_2,\dots,\bm{c}_{n},\Bar{\bm{c}}\cdot\Bar{\bm{a}}).$
    
\end{definition}
Define, $ \mathcal{T}^l_{\Bar{\bm{a}}},\, \widetilde{\mathcal{T}}^l_{\Bar{\bm{a}}}:\mathbb{F}_q^{nl}\rightarrow\mathbb{F}_q^{nl}$ such that, 
\begin{multline*}
    \mathcal{T}^l_{\Bar{\bm{a}}}(r_{1,1},r_{1,2},\dots,r_{1,l},r_{2,1},r_{2,2},\dots,r_{2,l},\dots,r_{n,1},\dots,r_{n,l})=(0,\dots,0,r_{1,1},r_{1,2},\dots,r_{1,l},\dots,\\r_{n-1,1},\dots,r_{n-1,l})+(r_{n,1}\,a_{1,1},\dots,r_{n,l}\,a_{n,l},r_{n,1}\,a_{2,1},\dots,r_{n,l}\,a_{2,l},\dots,r_{n,1}\,a_{n,1},\dots,r_{n,l}\,a_{n,l})
\end{multline*}
 and 
\begin{align*}
\widetilde{\mathcal{T}}^{l}_{\Bar{\bm{a}}}(r_{1,1},r_{1,2},\dots,r_{1,l},r_{2,1},r_{2,2},\dots,r_{2,l},\dots,r_{n,1},\dots,r_{n,l})=&(r_{2,1},r_{2,2},\dots,r_{2,l},\dots,r_{n,1},\\&\dots,r_{n,l},\,\bm{r}^1\cdot\bm{a}^{(1)},\bm{r}^2\cdot\bm{a}^{(2)},\dots,\bm{r}^l\cdot\bm{a}^{(l)})
\end{align*}
 Let $\Bar{\bm{a}}=(\bm{a}_1,\bm{a}_2,\dots,\bm{a}_{n})=\bm{e}_1*\bm{a}^{(1)}+\bm{e}_2*\bm{a}^{(2)}+ \dots +\bm{e}_l*\bm{a}^{(l)} \in (\mathbb{F}_q^l)^n$ and $\mathbf{a}_1$ is unit in $\mathbb{F}_q^l.$ 
\begin{definition}[$\Bar{\bm{a}}$-quasi cyclic code] A linear code $C$ of length $nl$ over $\mathbb{F}_q$ is said to be $\Bar{\bm{a}}$-quasi cyclic code of index $l$ over $\mathbb{F}_q$ if for any, $(r_{1,1},r_{1,2},\dots,r_{1,l},r_{2,1},r_{2,2},\dots,r_{2,l},\dots,r_{n,1},\dots,r_{n,l})\in C\Longrightarrow\mathcal{T}^l_{\Bar{\bm{a}}}(r_{1,1},r_{1,2},\dots,r_{1,l},r_{2,1},r_{2,2},\dots,r_{2,l},\dots,r_{n,1},\dots,r_{n,l})\in C.$
\end{definition}
\begin{definition}[$\Bar{\bm{a}}$-quasi sequential code] A linear code $C$ of length $nl$ over $\mathbb{F}_q$ is said to be $\Bar{\bm{a}}$-quasi sequential code of index $l$ over $\mathbb{F}_q$ if for any, $(r_{1,1},r_{1,2},\dots,r_{1,l},r_{2,1},r_{2,2},\dots,r_{2,l},\dots,r_{n,1},\\ \dots,r_{n,l})\in C\Longrightarrow\widetilde{\mathcal{T}}^{l}_{\Bar{\bm{a}}}(r_{1,1},r_{1,2},\dots,r_{1,l},r_{2,1},r_{2,2},\dots,r_{2,l},\dots,r_{n,1},\dots,r_{n,l})\in C.$   
\end{definition}
\begin{lemma}
    Let $\Bar{\bm{a}}=(\bm{a}_1,\bm{a}_2,\dots,\bm{a}_{n})=\bm{e}_1*\bm{a}^{(1)}+\bm{e}_2*\bm{a}^{(2)}+ \dots +\bm{e}_l*\bm{a}^{(l)} \in (\mathbb{F}_q^l)^n$ be such that  $\bm{a}_1$ is unit in $\mathbb{F}_q^l.$ Denote by $\tau_{\Bar{\bm{a}}},$ the $\Bar{\bm{a}}$-polycyclic shift.
    %and $\mathcal{T}^l_\mathbf{a}$ denotes the \ref{T^l} shift operator.
    Then, $\Phi\circ\tau_{\Bar{\bm{a}}}=\mathcal{T}^l_{\Bar{\bm{a}}}\circ\Phi.$  
\end{lemma}
\begin{proof}
Let $\Bar{\bm{c}}=(\bm{c}_1,\dots,\bm{c}_n)\in\mathcal{C}$ and $\Bar{\bm{a}}$ be as given in the theorem. Consider,
\begin{align*}
     \tau_{\Bar{\bm{a}}}(\Bar{\bm{c}})&=(\bm{0},\bm{c}_1,\dots,\bm{c}_{n-1})+\bm{c}_n(\bm{a}_1,\dots,\bm{a}_n)\\
     &=\left(\bm{0},\underset{i=1}{\overset{l}{\sum}}c_{1,i}\bm{e}_i,\dots,\underset{i=1}{\overset{l}{\sum}}c_{n-1,i}\bm{e}_i\right)+\underset{i=1}{\overset{l}{\sum}}c_{n,i}\bm{e}_i\left(\underset{i=1}{\overset{l}{\sum}}a_{1,i}\bm{e}_i,\dots,\underset{i=1}{\overset{l}{\sum}}a_{n,i}\bm{e}_i\right)\\
    &=\left(\bm{0},\underset{i=1}{\overset{l}{\sum}}c_{1,i}\bm{e}_i,\dots,\underset{i=1}{\overset{l}{\sum}}c_{n-1,i}\bm{e}_i\right)+\left(\underset{i=1}{\overset{l}{\sum}}c_{n,i}a_{1,i}\bm{e}_i,\dots,\underset{i=1}{\overset{l}{\sum}}c_{n,i}a_{n,i}\bm{e}_i\right)\\
    &=\left(\underset{i=1}{\overset{l}{\sum}}c_{n,i}a_{1,i}\bm{e}_i,\underset{i=1}{\overset{l}{\sum}}c_{1,i}\bm{e}_i+\underset{i=1}{\overset{l}{\sum}}c_{n,i}a_{2,i}\bm{e}_i,\dots,\underset{i=1}{\overset{l}{\sum}}c_{n-1,i}\bm{e}_i+\underset{i=1}{\overset{l}{\sum}}c_{n,i}a_{n,i}\bm{e}_i\right)\\
    &=\left(\underset{i=1}{\overset{l}{\sum}}c_{n,i}a_{1,i}\bm{e}_i,\underset{i=1}{\overset{l}{\sum}}(c_{1,i}+c_{n,i}a_{2,i})\bm{e}_i,\dots,\underset{i=1}{\overset{l}{\sum}}(c_{n-1,i}+c_{n,i}a_{n,i})\bm{e}_i\right)
\end{align*}  
Furthermore,\\
\begin{multline*}\label{Phitau}
    \Phi(\tau_{\Bar{\bm{a}}}(\Bar{\bm{c}}))=(c_{n,1}a_{1,1},c_{n,2}a_{1,2},\dots,c_{n,l}a_{1,l},c_{1,1}+c_{n,1}a_{2,1},c_{1,2}+c_{n,2}a_{2,2},\dots,c_{1,l}+c_{n,l}a_{n,l},\\\dots,c_{n-1,1}+c_{n,1}a_{n,1},\dots,c_{n-1,l}+c_{n,l}a_{n,l})
\end{multline*}
Now consider,
\begin{multline*}
    \Phi(\Bar{\bm{c}})=(c_{1,1},\dots,c_{1,l},c_{2,1},\dots,c_{2,l},\dots,c_{n,1},\dots,c_{n,l})\textit{ \textnormal{and} } \\
    \mathcal{T}^l_{\Bar{\bm{a}}}(\Phi(\Bar{\bm{c}}))=(0,\dots,0,c_{1,1},\dots,c_{1,l},\dots,c_{n-1,1},\dots,c_{n-1,l})+(c_{n,1}a_{1,1},\dots,c_{n,l}a_{n,l},c_{n,1}a_{2,1},\\\dots,c_{n,l}a_{2,l},\dots,c_{n,1}a_{n,1},\dots,c_{n,l}a_{n,l})
\end{multline*} \hfill $\square$
\end{proof}
Hence, the following theorem is immediate.
\begin{theorem}
    Let $\mathcal{C}$ be an $\mathbb{F}_q^l$-linear code having length $n$. Then, $\mathcal{C}$ is $\Bar{\bm{a}}$-polycyclic over $\mathbb{F}_q^l$ if and only if $\Phi(\mathcal{C})$ is $\Bar{\bm{a}}$-quasi cyclic code of index $l$ over $\mathbb{F}_q.$
\end{theorem}
\begin{lemma}
   Let $\Bar{\bm{a}}=(\bm{a}_1,\bm{a}_2,\dots,\bm{a}_{n})=\bm{e}_1*\bm{a}^{(1)}+\bm{e}_2*\bm{a}^{(2)}+ \dots +\bm{e}_l*\bm{a}^{(l)} \in (\mathbb{F}_q^l)^n$ be such that  $\bm{a}_1$ is unit in $\mathbb{F}_q^l.$ Denote by $\tau'_{\Bar{\bm{a}}},$ the $\Bar{\bm{a}}$-sequential shift. 
   Then, $\Phi\circ\tau'_{\Bar{\bm{a}}}=\widetilde{\mathcal{T}}^{l}_{\Bar{\bm{a}}}\circ\Phi.$ 
\end{lemma}
\begin{proof}
    Let $\Bar{\bm{c}}=(\bm{c}_1,\bm{c}_2,\dots,\bm{c}_n) \text{ and }\Bar{\bm{a}}=(\bm{a}_1,\bm{a}_2,\dots,\bm{a}_n) \in \mathbb{F}_q^{nl}.$ Then consider, 
    \begin{align*}
        \tau_{\Bar{\bm{a}}}'(\Bar{\bm{c}})&=(\bm{c}_2,\dots,\bm{c}_n,\Bar{\bm{c}}\cdot\Bar{\bm{a}})\\
       &=\left(\underset{j=1}{\overset{l}{\sum}}c_{2,j}\bm{e}_j,\dots,\underset{j=1}{\overset{l}{\sum}}c_{n,j}\bm{e}_j,\underset{j=1}{\overset{l}{\sum}}(c_{1,j}a_{1,j}+c_{2,j}a_{2,j}+\dots + c_{n,j}a_{n,j})\bm{e}_j\right)
    \end{align*}
    Then,\\
    \begin{align*}
        \Phi(\tau_{\Bar{\bm{a}}}'(\Bar{\bm{c}})) = & \big( c_{2,1},\dots,c_{2,l},c_{3,1},\dots,c_{3,l},\dots,c_{n,1},\dots,c_{n,l},c_{1,1}a_{1,1}+c_{2,1}a_{2,1}\\ & +\dots+c_{n,1}a_{n,1},c_{1,2}a_{1,2}+\dots+c_{n,2}a_{n,2},\dots,c_{1,l}a_{1,l}+c_{2,l}a_{2,l}+\dots+c_{n,l}a_{n,l}\big),
    \end{align*}
    which is clearly equal to $\widetilde{\mathcal{T}}^{l}_{\Bar{\bm{a}}}\circ\Phi(\Bar{\bm{c}}).$ \hfill $\square$
\end{proof}
Hence, the following theorem is immediate now:
\begin{theorem}
     Let $\mathcal{C}$ be an $\mathbb{F}_q^l$-linear code having length $n$. Then, $\mathcal{C}$ is $\Bar{\bm{a}}$-sequential code over $\mathbb{F}_q^l$ if and only if $\Phi(\mathcal{C})$ is $\Bar{\bm{a}}$-quasi sequential code of index $l$ over $\mathbb{F}_q.$
\end{theorem}
For the Gray map $\Phi$ defined above we know the algebraic structure on the corresponding Gray image of $\mathcal{C},$ however, to improve the distance we define generalized Gray map which is obtained by composing $\Phi$ with a suitable automorphism of $\mathbb{F}_q^{nl}$ which yields a code with the same length and dimension. \\
Hence, we consider the Gray map of the form:
\begin{equation}\label{GenGraymap}
\Psi(c_{1,1},c_{1,2},\dots,c_{1,l},\dots,c_{n,1},c_{n,2},\dots,c_{n,l})= ((c_{1,1},c_{1,2},\dots,c_{1,l})M,\dots, (c_{n,1},c_{n,2},\dots,c_{n,l})M),
\end{equation} where $M\in GL(l, \mathbb{F}_q).$\\

For Table \ref{table1}, $l=2$ and $M=\begin{pmatrix}
    1 & 1\\
    0 &1 \\
\end{pmatrix}.$ In Table \ref{table2} and \ref{table3}, $l$ takes value $3$ and $2$ respectively, and $M$ is as mentioned in each case. All codes presented in Table \ref{table1} and \ref{table2} (probably except item 31 of Table \ref{table1}) are optimal as per the database \cite{codetables} and new in the sense that they are inequivalent to the codes available in MAGMA \cite{cannon2011handbook}, whereas, in Table \ref{table3}, all codes are MDS.\\
\begin{lemma}
\label{annpreserving}
    Let $\mathcal{C}=\langle \Bar{\bm{g}}(x)=\underset{i=1}{\overset{l}{\sum}}\bm{e}_i\bm{g}^{(i)}(x)\rangle$ be an $\Bar{\bm{a}}$-polycyclic code of length $n$ over $\mathbb{F}_q^l$, and let $\Psi$ be the map defined as in Equation \eqref{GenGraymap}, then the parameters of $\Psi(\mathcal{C})$ will be given by $[nl,nl-\underset{i=1}{\overset{l}{\sum}}\deg {\bm{g}}^{(i)}(x),\geq d]_q,$ where $d$ is the Hamming distance of $\mathcal{C}$.
\end{lemma}
\begin{proof}
    The length and dimension of $\Psi(\mathcal{C})$ are clear from the definition of $\Psi.$ The distance remains greater than or equal to the Hamming distance of $\mathcal{C}$ since $M$ is an invertible matrix. 
\end{proof}\hfill $\square$\\
Observe that, the Gram matrix in Lemma \ref{Dual of dual} with respect to the annihilator product is given as:\\
 $$A=(
     \langle \Bar{\bm{\epsilon}}_i,\Bar{\bm{\epsilon}}_j \rangle_{\Bar{\bm{a}}})=\begin{pmatrix}
         (1,1,\dots,1)&\bm{0}&\bm{0}&\dots&\bm{0}\\
         \bm{0}&\bm{0}&\bm{0}&\dots&\bm{a}_0\\
         \bm{0}&\bm{0}&\bm{0}&\dots&\langle \Bar{\bm{\epsilon}}_3,\Bar{\bm{\epsilon}}_n \rangle_{\Bar{\bm{a}}}\\
         \vdots&\vdots&\vdots&\iddots&\vdots\\
         \bm{0}&\bm{a}_0&\langle \Bar{\bm{\epsilon}}_n,\Bar{\bm{\epsilon}}_3 \rangle_{\Bar{\bm{a}}}&\dots&\langle \Bar{\bm{\epsilon}}_n,\Bar{\bm{\epsilon}}_n \rangle_{\Bar{\bm{a}}}&
     \end{pmatrix}.
 $$Set $\langle \Bar{\bm{\epsilon}}_i,\Bar{\bm{\epsilon}}_j \rangle_{\Bar{\bm{a}}}:=(\epsilon_{i,j}^{(1)},\epsilon_{i,j}^{(2)},\dots,\epsilon_{i,j}^{(l)})\in \mathbb{F}_q^l.$ For every such $A,$ we can define $\Bar{A}\in GL(nl,\mathbb{F}_q)$ as follows:\\
  $$\Bar{A}=\begin{pmatrix}
     \begin{pmatrix}
         1&0&\dots&0\\
         0&1&\dots&0\\
         \vdots&\vdots&\vdots&\vdots\\
         0&0&\dots&1
     \end{pmatrix}&\begin{pmatrix}
         0&0&\dots&0\\
         0&0&\dots&0\\
         \vdots&\vdots&\vdots&\vdots\\
         0&0&\dots&0
     \end{pmatrix}&\dots&\begin{pmatrix}
         0&0&\dots&0\\
         0&0&\dots&0\\
         \vdots&\vdots&\vdots&\vdots\\
         0&0&\dots&0
     \end{pmatrix}\\

     \begin{pmatrix}
         0&0&\dots&0\\
         0&0&\dots&0\\
         \vdots&\vdots&\vdots&\vdots\\
         0&0&\dots&0
     \end{pmatrix}&\begin{pmatrix}
         0&0&\dots&0\\
         0&0&\dots&0\\
         \vdots&\vdots&\vdots&\vdots\\
         0&0&\dots&0
     \end{pmatrix}&\dots&\begin{pmatrix}
         a_0&0&\dots&0\\
         0&a_0&\dots&0\\
         \vdots&\vdots&\vdots&\vdots\\
         0&0&\dots&a_0
     \end{pmatrix}\\
     \vdots&\vdots&\iddots&\vdots\\

     \begin{pmatrix}
         0&0&\dots&0\\
         0&0&\dots&0\\
         \vdots&\vdots&\vdots&\vdots\\
         0&0&\dots&0
     \end{pmatrix}&\begin{pmatrix}
         a_0&0&\dots&0\\
         0&a_0&\dots&0\\
         \vdots&\vdots&\vdots&\vdots\\
         0&0&\dots&a_0
     \end{pmatrix}&\dots&\begin{pmatrix}
        \epsilon_{n,n}^{(1)}&0&\dots&0\\
         0&\epsilon_{n,n}^{(2)}&\dots&0\\
         \vdots&\vdots&\vdots&\vdots\\
         0&0&\dots&\epsilon_{n,n}^{(l)}
     \end{pmatrix}
 \end{pmatrix},$$ that is, $\Bar{A}=(A_{i,j}),$ where $A_{i,j}$ is $l\times l$ diagonal matrix $\textnormal{diag} (\langle \bm{\epsilon}_i,\bm{\epsilon}_j\rangle_{\Bar{\bm{a}}}).$ Note that if $\Bar{\bm{a}}=(\bm{a}_0,\bm{a}_1,\dots,\bm{a}_{n-1})\in(\mathbb{F}_q^l)^n$ is such that $\bm{a}_i=(b_i,b_i,\dots,b_i)\in \mathbb{F}_q^l$, for each $0\leq i \leq n-1,$ then every entry of the matrix $A=(\langle \bm{\epsilon}_i,\bm{\epsilon}_j\rangle_{\Bar{\bm{a}}})$ is again an $l$-tuple over $\mathbb{F}_q$ with equal components, that is, $\langle \bm{\epsilon}_i,\bm{\epsilon}_j \rangle_{\Bar{\bm{a}}}=(\epsilon_{i,j}^{(1)},\epsilon_{i,j}^{(2)},\dots, \epsilon_{i,j}^{(l)})=(\epsilon_{i,j}, \epsilon_{i,j},\dots, \epsilon_{i,j})\in \mathbb{F}_q^l.$ With these notations we prove the following Lemmas: 
\begin{lemma}\label{Psi(C.A)=Psi(C).Bar{A}}
    Let $\mathcal{C}$ be an $\Bar{\bm{a}}$-polycyclic code over $\mathbb{F}_q^l$ of length $n,$ and let $\Psi$ denote the map in Equation \eqref{GenGraymap}. If each $\bm{a}_i=(b_i, b_i, \dots, b_i)\in \mathbb{F}_q^l,$ where $b_i\in\mathbb{F}_q,$ then $$\Psi(\mathcal{C}.A)=\Psi(\mathcal{C}).\Bar{A}.$$
\end{lemma}
\begin{proof}
For any $\Bar{\bm{c}}=(\bm{c}_1,\bm{c}_2,\dots,\bm{c}_n)\in\mathcal{C},$ we get,
\begin{align*}
\Psi\big(\Bar{\bm{c}}A\big)=&\Big((\bm{c}_1(1,1,\dots,1))M\textbf{,}\,\,(\bm{c}_n\bm{a}_0)M\textbf{,}\,\, (\bm{c_{n-1}}\bm{a}_0)M+(\bm{c}_n\langle \bm{\epsilon}_n,\bm{\epsilon}_3 \rangle_{\Bar{\bm{a}}})M\textbf{,} \,\, \dots\,\,\textbf{,} (\bm{c}_2\bm{a}_0)M+(\bm{c}_3\\&\langle \bm{\epsilon}_3,\bm{\epsilon}_n \rangle_{\Bar{\bm{a}}})M+\dots+(\bm{c}_n\langle \bm{\epsilon}_n,\bm{\epsilon}_n \rangle_{\Bar{\bm{a}}})M \Big)\\=&\Big(\bm{c}_1M,\,b_0\bm{c}_nM,\, b_0\bm{c_{n-1}}M+\epsilon_{n,3}\,\bm{c}_n M, \dots, b_0\bm{c}_2M+\epsilon_{3,n}\,\bm{c}_3 M+\dots+\epsilon_{n,n}\,\bm{c}_nM\Big)\\=&(\bm{c}_1M,\,\bm{c}_2M,\dots,\bm{c}_nM)\Bar{A}\\=&
    \Psi(\Bar{\bm{c}})\Bar{A}
\end{align*}
Hence, the result follows.
 \hfill
    $\square$
\end{proof}
\begin{lemma}\label{same distance of Psi(C) and Psi(C).Bar(A)}
    If $\Bar{\bm{a}}(x)=\bm{a}_0\in U(\mathbb{F}_q^l),$ a constant polynomial, then $\min\{w_H(\Psi(\Bar{\bm{c}})):\Bar{\bm{c}}\in\mathcal{C}\setminus\{{\Bar{\bm{0}}}\}\}=\min\{w_H(\Psi(\Bar{\bm{c}})\Bar{A}):\Bar{\bm{c}}\in \mathcal{C}\setminus\{{\Bar{\bm{0}}}\}\}.$ Consequently, parameters of $\Psi(\mathcal{C})$ will be same as of $\Psi(\mathcal{C}).\Bar{A}.$
\end{lemma}
\begin{proof}
    The proof follows from the fact that $\Bar{A}$ is invertible with exactly one non-zero entry in each row and each column.\hfill
    $\square$
\end{proof}
\begin{theorem}[Euclidean CSS construction]\cite{calderbank1998quantum}\label{EuclcideanCSSconstruction}
Let $\mathcal{C}_1$ and $\mathcal{C}_2$ be $[n,k_1,d_1]_q$ and $[n,k_2,d_2]_q$ over $\mathbb{F}_q$ respectively with $\mathcal{C}_2^{\perp}\subseteq \mathcal{C}_1.$ Then there exists a quantum error-correcting code $\mathcal{C}$ with parameters $[[n,k_1+k_2-n,d]]_q,$ where $d=\min\{w_H(c): c\in(\mathcal{C}_1\setminus \mathcal{C}_2^\perp)\cup(\mathcal{C}_2\setminus \mathcal{C}_1^\perp)\}.$ 
\end{theorem}
Now, we have the following theorem.
\begin{theorem}\label{Quantumcons}
    Let $\mathcal{C}=\langle \Bar{\bm{g}}(x)\rangle$ be an $\Bar{\bm{a}}$-polycyclic code of length $n$ over $\mathbb{F}_q^l,$ where $\Bar{\bm{a}}(x)=\bm{a}_0=(b_0,b_0,\dots,b_0)\in\mathbb{F}_q^l,$ and let $\Psi$ be defined as in Equation \eqref{GenGraymap} and $MM^T=\lambda\, I$, for some $\lambda\in U(\mathbb{F}_q),$ then there exists a quantum code with parameters $[[nl,nl-2\underset{i=1}{\overset{l}{\sum}}\deg {\bm{g}}^{(i)}(x),\geq d]]_q,$ where $d$ is the Hamming distance of $\Psi(\mathcal{C}).$
\end{theorem}
\begin{proof}
    % Proof of the Theorem follows from the Lemma \ref{annpreserving} and the Theorem \ref{EuclcideanCSSconstruction}.
    From Lemma \ref{Psi(C.A)=Psi(C).Bar{A}}, we have that $\Psi(\mathcal{C}.A)=\Psi(\mathcal{C}).\Bar{A}$. Also $\Psi$ is an Euclidean duality preserving map, which follows from [Theorem 5.4, \cite{akanksha2025thetadeltathetamathbfacycliccodes}]. Then,$$\Psi(\mathcal{C}^\circ)=\Psi((\mathcal{C}.A)^\perp)=(\Psi(\mathcal{C}.A))^\perp=(\Psi(C).\Bar{A})^\perp.$$
    Hence, $\mathcal{C}^\circ\subset \mathcal{C}\implies \Psi(\mathcal{C}^\circ)\subset\Psi(\mathcal{C})\implies (\Psi(C).\Bar{A})^\perp\subset \Psi(\mathcal{C}).$
    Note that, since $\Bar{A}$ is invertible and by Lemma 
    \ref{same distance of Psi(C) and Psi(C).Bar(A)},
    $\Psi(\mathcal{C})$ and $\Psi(\mathcal{C}.\Bar{A})$ have same parameters. Hence, by Theorem \ref{EuclcideanCSSconstruction}, there exists a quantum code with parameter $[[nl, nl-2\underset{i=1}{\overset{l}{\sum}}\deg {\bm{g}}^{(i)}(x),\geq d]]_q,$ where $d$ is the Hamming distance of $\Psi(\mathcal{C}).$ \hfill $\square$
\end{proof}
An instance of this theorem can be seen in the following example: 
\begin{example}
    For $q=5$ and $l=2,$ consider the ring $\mathbb{F}_5^2.$ Let $x^n-\Bar{\bm{a}}(x)=(1,1)x^5-(1,1).$ Let $\mathcal{C}=\bm{e}_1\mathcal{C}_1\oplus\bm{e}_2\mathcal{C}_2$ be an $(1,1)$-polycyclic code of length $5$ over $\mathbb{F}_5^2.$ Then by Theorem \ref{polycyclic_classification} each $\mathcal{C}_i$ is a cyclic code over $\mathbb{F}_5$ of length $5$. Now,\\
    $$x^5-1=(x^2+3x+1)(x^3 + 2x^2 + 3x + 4)=(x+4)(x^4 + x^3 + x^2 + x + 1)$$
    If $\mathcal{C}_1=\langle x^2+3x+1 \rangle$, $\mathcal{C}_2=\langle x+4\rangle$ over $\mathbb{F}_5$ and $M=\begin{pmatrix}
        1&4\\
        4&4\\
    \end{pmatrix}.$ Then, $\Psi(\mathcal{C})$ will have parameters $[10,7,3].$ Moreover, $x^2+3x+1$ divides $x^3 + 2x^2 + 3x + 4$ and $x+4$ divides $x^4 + x^3 + x^2 + x + 1,$ $\mathcal{C}_1$ and $\mathcal{C}_2$ are both dual containing and consequently $\mathcal{C}$ is also dual containing by Corollary \ref{decomposition} and by Theorem \ref{Quantumcons}, there exists a quantum code with parameters $[[10,4 , \geq 3]].$ 
\end{example}
Table \ref{table4} consists of examples constructed using this construction.
\subsection{Examples For Tables}
\textbf{Example 1.} For $q=3$ and $l=3,$ consider the ring $\mathbb{F}_5^3.$ Let $x^n-\Bar{\bm{a}}(x)=(1,1,1)x^4-(2,2,2)x^3-(1,1,1)x-(1,1,1)\in \mathbb{F}_5^3[x].$ Suppose $\mathcal{C}=\bm{e}_1\mathcal{C}_1\oplus\bm{e}_2\mathcal{C}_2\oplus\bm{e}_3\mathcal{C}_3$ be an $((1,1,1),(1,1,1),(2,2,2))$-polycyclic code of length $4$ over $\mathbb{F}_5^3.$ Then by Theorem \ref{polycyclic_classification}, $\mathcal{C}_i$ is a $(1,1,2)$-polycyclic code of length $4$ over $\mathbb{F}_5$ for each $1\leq i \leq 3.$ Now,\\
    $$x^4-2x^3-x-1=\bm{g}^{(1)}(x)\times\bm{h}^{(1)}(x)=\bm{g}^{(2)}(x)\times\bm{h}^{(2)}(x)=\bm{g}^{(3)}(x)\times\bm{h}^{(3)}(x)$$
    in $\mathbb{F}_5[x].$ If $\mathcal{C}_1:=\langle\bm{g}^{(1)}(x)=x^3+2\rangle, \mathcal{C}_2:=\langle\bm{g}^{(2)}(x)=x^2+2\rangle$, $\mathcal{C}_3:=\langle\bm{g}^{(3)}(x)=x^3+2x^2+2x+1\rangle$  and $M=\begin{pmatrix}
    2&1&1\\
    1&2&1 \\
    0&1&1\\
    \end{pmatrix},$ then $\Psi(\mathcal{C})$ is an optimal $[12, 4, 6]_5$ linear code over $\mathbb{F}_5.$ This is example $5$ of Table \ref{table2}. All the above computations were done using MAGMA \cite{cannon2011handbook}.\\

    \textbf{Example 2.} For $q=7$ and $l=2,$ consider the ring $\mathbb{F}_7^2.$ Let $x^n-\Bar{\bm{a}}(x)=(1,1)x^3-(1,1)x-(1,1)\in \mathbb{F}_7^2[x].$ Suppose $\mathcal{C}=\bm{e}_1\mathcal{C}_1\oplus\bm{e}_2\mathcal{C}_2$ be an $((1,1),(1,1))$-polycyclic code of length $3$ over $\mathbb{F}_7^2.$ Then by Theorem \ref{polycyclic_classification}, $\mathcal{C}_i$ is a $(1,1)$-polycyclic code of length $3$ over $\mathbb{F}_7$ for each $1\leq i \leq 2.$ Now,\\
    $$x^3+6x+6=\bm{g}^{(1)}(x)\times\bm{h}^{(1)}(x)=\bm{g}^{(2)}(x)\times\bm{h}^{(2)}(x)$$
    in $\mathbb{F}_7[x].$ If $\mathcal{C}_1:=\langle\bm{g}^{(1)}(x)=x^2+5x+3\rangle, \mathcal{C}_2:=\langle\bm{g}^{(2)}(x)=x+2\rangle$ and $M=\begin{pmatrix}
    5&2\\
    2&2\\
    \end{pmatrix},$ then $\Psi(\mathcal{C})$ is an MDS $[6,3,4]_7$ linear code over $\mathbb{F}_7.$ This is example $4$ of Table \ref{table3}. All the above computations were done using MAGMA \cite{cannon2011handbook}.

\begin{landscape}

\begin{table}[!ht]

    \begin{adjustbox}{width=.9\textwidth}
   
    \begin{tabular}{|l|l|l|l|l|l|l|l|}
    \hline
       S. No.& $q$ & $n$ & $\Bar{\bm{g}}(x)=x^n-\Bar{\bm{a}}(x)$ & $\bm{g}^{(1)}(x)$ & $\bm{g}^{(2)}(x)$ &  Parameters of $\Psi(C)$ & Remarks \\ 
        \hline
       1    & 2 & 6 & $(1,1)x^6-(1,1)x^5-(1,1)x^2-(1,1)$ & $x^2+x+1$ & $x^4+x^2+x+1$ & $[12,6,4]$ & Optimal \\
        \hline
        2    &2 & 6&$(1,1)x^6 + (1,1)x^3 + (1,1)x + (1,1)$& $x+1$ & $x^3 + x^2 + 1$ &  $[12,8,3]$ & Optimal \\
        \hline
        3    &2 & 7&$(1,1)x^7-(1,1)$& $x + 1$ & $x^6 + x^5 + x^4 + x^3 + x^2 + x + 1$ &  $[14,7,4]$ & Optimal and Cyclic \\
        \hline
        4    &2 & 7&$(1,1)x^7-(1,1)$& $x^3 + x^2 + 1$ & $x^6 + x^5 + x^4 + x^3 + x^2 + x + 1$ &  $[14,5,6]$ & Optimal and Quasicyclic \\
        \hline
        5    &2 & 8&$(1,1)x^8 + (1,1)x^5 + (1,1)x^3 + (1,1)$& $x + 1$ & $x^6 + x^5 + x + 1$ &  $[16,9,4]$ & Optimal and Quasicyclic \\
        \hline
         6     &2 & 8&$(1,1)x^8 + (1,1)x^4 + (1,1)x^2 + (1,1)$& $x^5 + x^4 + x^3 + 1$ & $x^8 + x^4 + x^2 + 1$ &  $[16, 3, 8]$ & Optimal and Quasicyclic \\
        \hline
        7   &2 & 8&$(1,1)x^8 - (1,1)x^4 - (1,1)x^2 - (1,1)$& $x + 1$ & $x^5 + x^4 + x^3 + 1$ &  $[16,10,4]$ & Optimal and Quasicyclic \\
        \hline
        8   &2 & 8&$(1,1)x^8+(1,1)x^6+(1,1)x^2+(1,1)$& $x^5+x^3+x^2+1$ & $x^8+x^6+x^2+1$ &  $[16,3,8]$ & Optimal and Quasicyclic \\
        \hline
        9   &3 & 6 & $(1,1)x^6 + (2,2)x^2 + (2,2)x + (2,2)$ & $x + 1$ &$x^3 + 2x + 1$&$[12,8,3]$  & Optimal, LCD \\
       % \hline
       % 4&3&12&$x^6-x^2-x-1$&$x^2+2x+2$&$1$&$[12,10,2]$ &Optimal and  Almost MDS \\
        \hline
        
       10    &3 & 6 & $(1,1)x^6+(2,2)x^2+(2,2)x+(2,2)$& $x+1$ & $x^4+x^3+2x^2+1$ & $[12,7,4]$ & Optimal and Quasicyclic\\
        \hline
        11    &3 & 7 & $(1,1)x^7 + (1,1)x^4 + (2,2)x + (2,2)$&$x + 1$&$x^3 + x^2 + 2$&$[14,10,3]$&Optimal \\
    
        \hline
        12   &3&7&$(1,1)x^7-(2,2)x^4-(1,1)x-(1,1)$&$x+2$&$x^4+2x^2+2x+1$&$[14,9,4]$&Optimal and Quasicyclic\\
        \hline
       13   &3&8&$(1,1)x^8-(2,2)x^4-(1,1)x^3-(1,1)$&$x+2$&$x^3+2x+2$&$[16,12,3]$&Optimal, LCD\\
        \hline
        14   &3&8&$(1,1)x^8 + (1,1)x^3 + (2,2)x + (2,2)$&$x+1$&$x^5+2x^4+2x^3+x^2+1$&$[16,10,4]$&Optimal \\
        \hline
       15   &3&8&$(1,1)x^8-(2,2)x^3-(1,1)x-(1,1)$&$x^5+2x^4+2x^3+x^2+1$&$x^8+x^3+2x+2$&$[16,3,10]$&Optimal and Quasicyclic\\
       \hline
       16   &3&8&$(1,1)x^8-(2,2)x^3-(1,1)x-(1,1)$&$x^6+x^5+2x^3+2x^2+x+2$&$x^8+x^3+2x+2$&$[16,2,12]$&Optimal and Quasicyclic \\
       \hline
       17   &4 & 4 & $(1,1)x^4-(u,u)x^3-(u,u)x^2-(1,1)$ & $x+1$ &$x^2+x+u^2$&$[8,5,3]$  & Optimal \\
        \hline
       18    &4& 4&$(1,1)x^4+(u,u)x^3+(1,1)$&$x+u^2$&$x^3+x^2+u^2x+u$&$[8,4,4]$&Optimal and A-MDS \\
       \hline
       19    &4&  5& $(1,1)x^5+(u,u)x^2+(u,u)x+(1,1)$&$x+1$&$x^2+x+u$&$[10,7,3]$&Optimal \\
       \hline
       20   &4&5&$(1,1)x^5+(u,u)x^2+(u,u)x+(1,1)$&$x^2+x+u$&$x^5+ux^2+ux+1$&$[10,3,6]$&Optimal and Quasicyclic \\
       \hline
       21    &4&5&$(1,1)x^5+(1,1)$&$x+1$&$x^3+u^2x^2+u^2x+1$&$[10,6,4]$&Optimal, Quasicyclic, and A-MDS \\
       \hline
        22    &4&6&$(1,1)x^6+(u,u)x^2+(u,u)x+(1,1)$&$x+1$&$x^4+u^2x^3+x^2+x+u$&$[12,7,4]$&Optimal and Quasicyclic \\
       \hline
       23     &4&7&$(1,1)x^7+(1,1)x^4+(u,u)x+(1,1)$&$x+u^2$&$x^4+x^3+ux^2+x+1$&$[14,9,4]$&Optimal \\
       \hline
        
        24    &5 & 4& $(1,1)x^4-(1,1)$& $x+1$ & $x^3+4x^2+x+4$ & $[8,4,4]$ & Optimal, LCD and A-MDS \\
        \hline
        25   &5 & 4&$(1,1)x^4-(1,1)x^3-(1,1)x-(1,1)$& $x+3$ & $x^2 +4x + 4$ & $[8,5,3]$ & Optimal and A-MDS \\
        \hline
        26   &5 & 5&$(1,1)x^5-(1,1)$& $x-1$ & $(x-1)^3$ &  $[10,6,4]$ & Optimal and A-MDS \\
        \hline
        27   &5 & 5 &$(1,1)x^5-(1,1)$ & $x-1$ & $(x-1)^2$ &  $[10,7,3]$ & Optimal and A-MDS \\
        \hline
        28   &5 & 5 & $(1,1)x^5+(1,1)x^4+(1,1)x^3+(2,2)x^2+(4,4)$ & $x^3 + 3x^2 + 3x + 1$ & $x^5 + x^4 + x^3 + 2x^2 + 4$ &  $[10,2,8]$ & Optimal and A-MDS \\
        %\hline
        %23&5 & 10 & $x^5+1$ & $x^2+2x+1$ & $1$ &  $[10,8,2]$ & Optimal and Almost MDS \\
        \hline
        %29&7 & 8 & $x^4+2x^2+4$ & $x^4+2x^2+4$ & $x^2+x+5$ &  $[8,2,6]$ & Almost Optimal and Almost MDS \\
        %\hline

         29   &7 & 5 & $(1,1)x^5-(2,2)x^4-(1,1)x-(1,1)$ & $x+2$ & $x^3+6x^2+5x+6$ &  $[10,6,4]$ & Optimal, LCD and A-MDS \\
       % \hline
         %31&7 & 10 & $x^5+5x^3+6x+6$ & $x^5+5x^3+6x+6$ & $x^2+6x+6$ &  $[10,3,6]$ & Almost Optimal and Quasicyclic \\
        %\hline
      % 32&7 & 10 & $x^5-2x^4-x-1$ & $x+4$ & $(x+2)(x^3+6x^2+5x+6)$ &  $[10,5,4]$ & Almost Optimal \\
         \hline
          30   &7 & 5 & $(1,1)x^5+(5,5)x^3+(3,3)x+(6,6)$ & $x^3+4x^2+6x+6$ & $x^5+5x^3+3x+6$ &  $[10,2,8]$ & Optimal, LCD, Quasicyclic and A-MDS \\
        \hline
        31    &11 & 4 & $(1,1)x^4+(5,5)x^2+(6,6)x+(2,2)$ & $x+3$ & $x^2+5x+10$ &  $[8,5,3]$ & LCD, A-MDS (No data available) \\
        \hline
        \multicolumn{7}{l}{\normalsize In Table
        A-MDS stands for Almost MDS, that is, 
       $k+d=n$.}
    \end{tabular}
    
    \end{adjustbox}
   \caption{$l=2$}
    \Large\label{table1}
\end{table}
\newpage
\centering

\begin{table}[!ht]
    \begin{adjustbox}{width=0.8\textwidth}
    
    \begin{tabular}{|l|l|l|l|l|l|l|l|l|l|}
    \hline
       S. No.& $q$ & $n$ &$M$& $\Bar{\bm{g}}(x)=x^n-\Bar{\bm{a}}(x)$ & $\bm{g}^{(1)}(x)$ & $\bm{g}^{(2)}(x)$ &$\bm{g}^{(3)}(x)$&  Parameters $\Psi(C)$ & Remarks \\ 
       
        \hline
        
        1    &2&5   &$\begin{pmatrix}
            1&1&1\\
            0&1&1 \\
            1&0&1  
        \end{pmatrix}$ & $(1,1,1)x^5+(1,1,1)$ & $x+1$& $x^4+x^3+x^2+x+1$&$x^4+x^3+x^2+x+1$&$[15,6,6]$&Optimal, LCD and Quasicyclic\\
        \hline
       2          &2&6&$\begin{pmatrix}
           1&1&1\\
           0&1&1 \\
           1&0&1  
        \end{pmatrix}$ & $(1,1,1)x^6+(1,1,1)$ & $x^4+x^3+x+1$& $x+1$&$x+1$&$[18,12,4]$&Optimal and Quasicyclic\\
        \hline
        3   &2&7   &$\begin{pmatrix}
            1&1&1\\
            0&1&1 \\
            1&0&1  
        \end{pmatrix}$ & $(1,1,1)x^7+(1,1,1)$ & $x^4+x^3+x^2+1$& $x+1$&$x+1$&$[21,15,4]$&Optimal and Quasicyclic\\
        \hline
       
        4    &3&5               &$\begin{pmatrix}
            1&1&1\\
            0&2&1 \\
            0&1&1  
        \end{pmatrix}$ & $(1,1,1)x^5-(1,1,1)$ & $x-1$& $x-1$&$x^4+x^3+x^2+x+1$&$[15,9,4]$&Optimal \\
        \hline
        5   &3&4               &$\begin{pmatrix}
            2&1&1\\
            1&2&1 \\
            0&1&1  
        \end{pmatrix}$ & $(1,1,1)x^4-(2,2,2)x^3-(1,1,1)x-(1,1,1)$ & $x^3+2$& $x^2+2$&$x^3+2x^2+2x+1$&$[12,4,6]$&Optimal \\
        \hline
         6    &3&4              &$\begin{pmatrix}
            2&1&1\\
            1&2&1 \\
            0&1&1  
        \end{pmatrix}$ & $(1,1,1)x^4 + (1,1,1)x^3 + (2,2,2)x + (2,2,2)$ & $x^3 + 2x^2 + 2x + 1$& $x+2$&$x+1$&$[12,7,4]$&Optimal \\
        \hline
         7    &3&4            &$\begin{pmatrix}
            2&1&1\\
            1&2&1 \\
            0&1&1  
        \end{pmatrix}$ & $(1,1,1)x^4+(1,1,1)x^3+(2,2,2)x+(2,2,2)$ & $x^2+x+1$& $x+1$&$x+1$&$[12,8,3]$&Optimal \\
        \hline
        8     &4&3           &$\begin{pmatrix}
            u^2&0&u\\
            u&u^2&1 \\
            1&1&u  
        \end{pmatrix}$ & $(1,1,1)x^3+(1,1,1)$ & $x+u^2$& $x+1$&$x+u^2$&$[9,6,3]$&Optimal, LCD \\
        \hline
        9     &5&3            &$\begin{pmatrix}
            1&0&1\\
            0&1&0 \\
            2&2&1  
        \end{pmatrix}$ & $(1,1,1)x^3+(1,1,1)x^2+(4,4,4)$ & $x+2$& $x^2+4x+2$&$1$&$[9,6,3]$&Optimal, LCD and A-MDS \\
        \hline
        %6&5&12&$\begin{pmatrix}
            %1&0&1\\
           % 0&1&0 \\
           % 2&2&1  
       % \end{pmatrix}$ & $x^4+x^2+4$ & $1$& $x^2+3$&$1$&$[12,10,2]$&Optimal and Almost MDS \\
       % \hline
        10      &7&3       &$\begin{pmatrix}
            1&0&1\\
            0&1&0 \\
            2&2&1  
        \end{pmatrix}$ & $(1,1,1)x^3+(4,4,4)x+(4,4,4)$ & $x+3$& $x^2+4x+6$&$1$&$[9,6,3]$&Optimal, LCD and A-MDS \\
        \hline
         11     &7&3          &$\begin{pmatrix}
            1&0&1\\
            1&1&0 \\
            0&1&1  
        \end{pmatrix}$ & $(1,1,1)x^3-(6,6,6)x^2-(5,5,5)x-(9,9,9)$ & $x^2+3x+1$& $x+5$&$x+5$&$[9,5,4]$&Optimal and A-MDS \\
        \hline
        12      &7&3           &$\begin{pmatrix}
            1&5&1\\
            1&1&0 \\
            2&1&1  
        \end{pmatrix}$ & $(1,1,1)x^3-(6,6,6)x^2-(5,5,5)x-(9,9,9)$ & $x^3+x^2+2x+5$& $x^2+3x+1$&$x+5$&$[9,3,6]$&Optimal, LCD and A-MDS \\
        \hline
        \multicolumn{7}{l}{\normalsize In Table
        A-MDS stands for Almost MDS, that is, 
       $k+d=n$.}
        \end{tabular}
    
    \end{adjustbox}
   \caption{$l=3$}
    \Large\label{table2}
\end{table} 

\end{landscape}
Apart from these codes, we are able to obtain some MDS codes.
{\Large{
\begin{table}[!ht]
    \centering
    \begin{adjustbox}{width=.9\textwidth}
    \begin{tabular}{|l|l|l|l|l|l|l|l|l|}
    \hline
   S.No.&$q$ & $n$ & $M$& $\Bar{\bm{g}}(x)=x^n-\Bar{\bm{a}}(x)$ & $\bm{g}^{(1)}(x)$ & $\bm{g}^{(2)}(x)$ &  Parameters $\Psi(C)$ & Remark \\ 
   \hline
       1   &4       &3         &$\begin{pmatrix}
            u&u^2\\
            u&u \\
        \end{pmatrix}$&$(1,1)x^3+(1,1)$&$x+u$&$x^2+ux+u^2$&$[6,3,4]$&LCD, MDS\\
   \hline
     2   &5       &3            &$\begin{pmatrix}
            3&2\\
            2&2 \\
        \end{pmatrix}$&$(1,1)x^3+(4,4)$&$x^2+x+1$&$x+4$&$[6,3,4]$&LCD, MDS\\
   \hline
    3    &7         &3           &$\begin{pmatrix}
            5&2\\
            2&2 \\
        \end{pmatrix}$&$(1,1)x^3+(6,6)$&$x+6$&$x+3$&$[6,4,3]$&MDS\\
   \hline
       4    &7      &3            &$\begin{pmatrix}
            5&2\\
            2&2 \\
        \end{pmatrix}$&$(1,1)x^3+(6,6)x+(6,6)$&$x^2+5x+3$&$x+2$&$[6,3,4]$&MDS\\
   \hline
       %14&7&6&$\begin{pmatrix}
        %    5&2\\
         %   2&2 \\
        %\end{pmatrix}$&$x^3+5x+6$&$x+1$&$1$&$[6,5,2]$&Optimal\\
   %\hline
      5   &7       &3           &$\begin{pmatrix}
            5&2\\
            2&2 \\
        \end{pmatrix}$&$(1,1)x^3+(6,6)$&$x^2+4x+2$&$x^2+2x+4$&$[6,2,5]$&MDS\\
   \hline
    6    &8       &2     &$\begin{pmatrix}
            u&u^3\\
            1&u \\
        \end{pmatrix}$&$(1,1)x^2+(u^2,u^2)x+(u,u)$&$x+u^3$&$x+u^5$&$[4,2,3]$&LCD, MDS\\
   \hline
    7    &8          &3         &$\begin{pmatrix}
            u&u^3\\
            u&u \\
        \end{pmatrix}$&$(1,1)x^3+(u^2,u^2)x+(u,u)$&$x+u^4$&$x^2+u^4x+u^4$&$[6,3,4]$&LCD, MDS\\
   \hline
    8    &9         &4        &$\begin{pmatrix}
            u&u\\
            u^2&1 \\
        \end{pmatrix}$&$(1,1)x^4 + (2,2)$&$x^2+u^3x + u^2$&$x^3 + u^2x^2+2x + u^6$&$[8, 3, 6]$& LCD,  MDS\\
   \hline
   9   &9          &4        &$\begin{pmatrix}
            u&u\\
            u^2&1 \\
        \end{pmatrix}$&$(1,1)x^4 + (2,2)$&$x + u^2$&$x^2 + u^5x + u^6$&$[8, 5, 4]$&LCD, MDS\\
        \hline
    \end{tabular}
    
    \end{adjustbox}
    \caption{$l=2$}
    \label{table3}
\end{table}}}
\label{Graymap}

{\Large{
\begin{table}[!ht]
    \centering
    \begin{adjustbox}{width=.9\textwidth}
    \begin{tabular}{|l|l|l|l|l|l|l|l|l|l|}
    \hline
   $q$ & $n$ & $M$& $\Bar{\bm{g}}(x)=x^n-\Bar{\bm{a}}(x)$ & $\bm{g}^{(1)}(x)$ & $\bm{g}^{(2)}(x)$ &  Parameters $\Psi(C)$&Remarks about $\Psi(C)$ &  Parameters of Quantum Code \\ 
   \hline
    3       &9         &$\begin{pmatrix}
            2&1\\
            2&2 \\
        \end{pmatrix}$&$(1,1)x^9+(2,2)$&$x+2$&$x^4+2x^3+2x+1$&$[18,13,3]$&Almost optimal&$[[18,8,\geq 3]]$\\
   \hline
    3       &9         &$\begin{pmatrix}
            1&0\\
            0&1 \\
        \end{pmatrix}$&$(1,1)x^9+(2,2)$&$x+2$&$x+2$&$[18,16,2]$&Optimal&$[[18,14,2]]$\\
   \hline
    5       &5        &$\begin{pmatrix}
            1&4\\
            4&4 \\
        \end{pmatrix}$&$(1,1)x^5-(1,1)$&$x^2+3x+1$&$x+4$&$[10,7,3]$&Optimal&$[[10,4,\geq 3]]$\\
   \hline
    7       &7        &$\begin{pmatrix}
            2&4\\
            3&2 \\
        \end{pmatrix}$&$(1,1)x^7-(1,1)$&$x^2+5x+1$&$x^2+5x+1$&$[14,10,3]$&Almost Optimal&$[[14,6,\geq 3]]$\\
   \hline
   7       &7        &$\begin{pmatrix}
            2&4\\
            3&2 \\
        \end{pmatrix}$&$(1,1)x^7-(1,1)$&$x^2+5x+1$&$x^3+4x^2+3x+6$&$[14,9,4]$&Almost Optimal&$[[14,4,\geq 4]]$\\
   \hline
    
    \end{tabular}
    
    \end{adjustbox}
    \caption{Quantum Codes}
    \label{table4}
\end{table}}}

\section{Conclusion}\label{sec6}
To study polycyclic (in particular, cyclic) codes, many authors utilized $\mathbb{F}_q$-algebras that have an orthogonal basis of idempotents having sum 1 and they could extract interesting codes, for instance, Qi (\cite{qi2022polycyclic}) constructed almost MDS binary codes,  Islam et al. (\cite{islam2021cyclic}) produced many optimal and MDS codes, and Bhardwaj et al. (\cite{swati_raka}) considered a broader class of rings generalizing the above two and studied constacyclic codes over them. Since the algebras considered by them are each isomorphic to $\mathbb{F}_q^l$ for some $l\in\mathbb{N},$ we have studied polycyclic codes over the product ring $\mathbb{F}_q^l$.   
In fact, in these articles, the product ring structure of the base ring was utilized to extract good codes. For instance, in \cite{islam2021cyclic}, they picked a base ring which is isomphic to $\mathbb{F}_q[u]/\langle u^l-1\rangle$ so that $u^l-1$ splits over $\mathbb{F}_q$; in particular, it requires, $l\leq q$ and $l\mid q-1.$ Note that the product ring $\mathbb{F}_q^l$ cannot be realized as $\mathbb{F}_q[u]/\langle f(u)\rangle$ if $l>q.$ 
Our study of  polycyclic codes over $\mathbb{F}_q^l$ does not assume any  restriction on $q$ and $l$.  For example, binary codes cannot be obtained as the Gray image by the construction in \cite{islam2021cyclic}. Moreover, for $l=2,$ linear codes (for instance, Items 3, 4, 21 in Table 1, and 1 in Table 3) over fields with characteristic $2$ cannot be obtained by the above mentioned constructions. In fact, our general setup enables us to construct a large number of interesting codes, since $q$ and $l$ are independent of each other. As a future work, one can try to find a characterization of Gray maps that can produce good codes; in other words, one can try to find a class of $M$ in Section \ref{Graymap} which yields codes with good parameters. 

%\subsection*{Acknowledgements}

\subsection*{Declarations}
{\bf Ethical Approval and Consent to participate:} Not applicable.\\
{\bf Consent for publication:} Not applicable.\\
{\bf Availability for supporting data:} Not applicable.\\
{\bf Competing interests:} The authors declare that they have no competing
interests.\\
{\bf Funding:} The first author expresses gratitude to MHRD, India, for financial support in the form of PMRF(PMRF ID 1403187) at the Indian Institute of Technology, Delhi.\\
{\bf Author’s contributions:} Akanksha and Ritumoni Sarma contributed
equally to this work. Both authors read and approved the final manuscript.\\
{\bf Acknowledgments:} We thank FIST lab (project number: SR/FST/MS-1/2019/45) for computation facility.
\bibliographystyle{abbrv}
\bibliography{Genproductring}

\nocite{huffman2010fundamentals}
\nocite{ling2004coding}
\nocite{cannon2006handbook}
%\printbibliography

\end{document}